\documentclass[final]{siamltex}
\usepackage{amsmath, amssymb, amsfonts}
\usepackage{graphicx,color}
\usepackage{diagbox}
\usepackage{algorithm}
\usepackage{algpseudocode}
\usepackage{enumitem}
\usepackage{multirow}
\usepackage{subfig}
\usepackage{epstopdf}
\usepackage{comment}
\usepackage{adjustbox}

\newcommand{\Tr}{\mathrm{Tr}}

\newcommand{\mc}[1]{\mathcal{#1}}

\newcommand{\cE}{\mathcal{E}}

\newcommand{\ud}{\,\mathrm{d}}

\newcommand{\Hxc}{\mathrm{Hxc}}

\newcommand{\abs}[1]{\lvert#1\rvert}
\newcommand{\norm}[1]{\lVert#1\rVert}

\newcommand{\ie}{\textit{i.e.},~}
\newcommand{\eg}{\textit{e.g.},~}
\newcommand{\Or}{\mathcal{O}}

\newcommand{\I}{\imath}

\newcommand{\RR}{\mathbb{R}}
\newcommand{\CC}{\mathbb{C}}

\global\long\def\Tr{\mathrm{Tr}}

\numberwithin{equation}{section}
\numberwithin{figure}{section}
\newtheorem{thm}{\protect\theoremname}

\newtheorem{rem}[thm]{\protect\remarkname}
\newtheorem{prop}[thm]{\protect\propositionname}

\providecommand{\corollaryname}{Corollary}
\providecommand{\lemmaname}{Lemma}
\providecommand{\propositionname}{Proposition}
\providecommand{\remarkname}{Remark}
\providecommand{\theoremname}{Theorem}

\newcommand{\REV}[1]{\textcolor{black}{#1}}
\newcommand{\CS}{\mathcal{C}}

\title{Projection based embedding theory for solving Kohn-Sham density functional theory}
\author{
Lin Lin\thanks{Department of Mathematics, University of California,
Berkeley, Berkeley, CA 94720 and Computational Research Division,
Lawrence Berkeley National Laboratory, Berkeley, CA 94720. Email:
\texttt{linlin@math.berkeley.edu}}\and Leonardo
Zepeda-N\'u\~nez\thanks{Computational Research Division, Lawrence Berkeley
National Laboratory, Berkeley, CA 94720. Email: \texttt{lzepeda@lbl.gov}}}

\begin{document}
\maketitle
\begin{abstract}
  Quantum embedding theories are playing an increasingly important role
  in bridging different levels of approximation to the many-body
  Schr\"odinger equation in physics, chemistry and materials science.
  In this paper, we present a linear algebra perspective of the recently
  developed projection based embedding theory (PET) [Manby et al,  J.
  Chem. Theory Comput. 8, 2564, 2012], restricted to the context of Kohn-Sham
  density functional theory.  By partitioning the global degrees of freedom
  into a ``system'' part and a ``bath'' part, and by
  choosing a proper projector from the bath, PET is an {\it in-principle}
  exact formulation to confine the calculation to the system part
  only, and hence can be performed with reduced computational cost.
  Viewed from the perspective of domain decomposition methods, one
  particularly interesting feature of PET is that it does not enforce a
  boundary condition explicitly, and remains applicable even when the
  discretized Hamiltonian matrix is dense, such as in the context of the
  planewave discretization. In practice, the accuracy of PET depends on
  the accuracy of the \REV{bath projector}. Based on the linear
  algebra reformulation, we develop a first order perturbation
  correction to the projector from the bath to improve its accuracy.
  Numerical results for real chemical systems indicate that with a
  proper choice of reference system \REV{used to compute the bath projector},
  the perturbatively corrected PET
  can be sufficiently accurate even when strong perturbation is applied
  to very small systems, such as the computation of the ground state
  energy of a SiH$_3$F molecule, using a SiH$_4$ molecule as the
  reference system.
 % \REV{ I guess that we would need to clarify how a reference system would play a role.}
\end{abstract}

\section{Introduction}

Multiphysics simulation usually involves two or more physical scales. In
the context of electronic structure theory, even though everything \REV{
on the scale of electrons and molecules
is described by the many-body Schr\"odinger equation,}
the concept \REV{behind} multiphysics
simulation remains valid. The direct solution to the
many-body Schr\"odinger equation itself is prohibitively expensive,
except for systems with a handful of electrons. This has led to the development
of various theoretical tools in both quantum physics and quantum chemistry to
find approximate solutions to the Schr\"odinger equation. These theories  can be
effectively treated as ``different levels of physics'' providing different levels of
accuracy. However, depending on the accuracy required, the
computational cost associated with such approximate theories can still
be very high.  So if a large quantum system can be partitioned into a
``system'' part containing the degrees of freedom that
are of interest that need to be treated using a \REV{relatively} accurate theory,
and a ``bath'' part containing the rest of the
degrees of freedom that can be treated using a less accurate theory, it becomes
naturally desirable to have a  numerical method that can
bridge the two levels of theories. In quantum physics, such
``multiscale'' methods have been actively developed in the past few
decades and are often called  ``quantum embedding theories''
(see \eg~\cite{BernholcLipariPantelides1978,ZellerDederichs1979,Cortona1991,KniziaChan2013,GoodpasterAnanthManbyEtAl2010,HuangPavoneCarter2011,GarciaLuE2007,ChenOrtner2015,WilliamsFeibelmanLang1982,KellyCar1992,
ZgidChan2011,KananenkaGullZgid2015,NguyenKananenkaZgid2016,ChibaniRenSchefflerEtAl2016,LiLinLu2018}
and ~\cite{SunChan2016} for a recent brief review).

The projection based embedding theory
(PET)~\cite{ManbyStellaGoodpasterEtAl2012} is a recently developed
quantum embedding theory, which is a versatile method that can be used
to couple a number of quantum theories together in a seamless fashion
(also see recent
works~\cite{LibischMarsmanBurgdorferEtAl2017,ChulhaiGoodpaster2018}).
This paper is a first step towards a mathematical understanding of PET. To make the discussions concrete, we assume that the system
part is described by the widely used Kohn-Sham density functional theory
(KSDFT)~\cite{HohenbergKohn1964,KohnSham1965}, and the bath part is also
described by KSDFT but solved only approximately.  Although this
setup is simpler than \REV{the one presented} in~\cite{ManbyStellaGoodpasterEtAl2012}, it
is already interesting from the perspective of approximate solution of large scale
eigenvalue problems, as to be detailed below.

After proper discretization, KSDFT can be written as the
following nonlinear eigenvalue problem
\begin{equation}\label{eqn:KSDFT}
  H[P] \Psi = \Psi \Lambda, \quad P=\Psi\Psi^{*},
\end{equation}
where \REV{the Hamiltonian} $H[P]\in \CC^{N \times N}$ is a Hermitian matrix,
and the diagonal matrix $\Lambda\in \RR^{N_{e}\times N_{e}}$ encodes the
algebraically lowest $N_e$ eigenvalues ($N\gg N_{e}$). $N$ is the number of
degrees of freedom of the Hamiltonian operator after discretization,
and $N_{e}$ is the number of electrons in the system (spin degrees of
freedom omitted). The eigenvectors associated with $\Lambda$ are denoted
by $\Psi=[\psi_{1},\ldots,\psi_{N_{e}}]\in \CC^{N \times N_{e}}$, and
$\Psi$ satisfies the orthonormality condition $\Psi^{*}\Psi=I_{N_{e}}$,
\REV{where $I_{N_{e}}$ is the identity of size $N_e$}.
The matrix $P$ is a spectral projector, usually called the {\it density
matrix}.  The Hamiltonian $H[P]$ depends
on the density matrix $P$ in a nonlinear fashion, and Eq.~\eqref{eqn:KSDFT} needs to
be solved self-consistently.

Without loss of generality, the system part can be defined
as the degrees of freedom associated with a set of indices
$\mc{I}_{s}$, and the bath part with a set of indices $\mc{I}_{b}$, so
that $\mc{I}_{s}\cup \mc{I}_{b}=\{1,\ldots,N\}$. \REV{Usually
$|\mc{I}_{b}|\gg |\mc{I}_{s}|$.} We are mostly interested in
the accurate computation of physical observables associated with the
system part, \ie the matrix block of the density matrix
$P_{\mc{I}_{s},\mc{I}_{s}}$. Since the eigenvalue
problem~\eqref{eqn:KSDFT} couples all degrees of freedom together,
this task still requires a relatively accurate description of the rest
of the density matrix.

In a nutshell, PET assumes the following decomposition of the density matrix
\begin{equation}
  P=P_{s}+P_{0,b},
  \label{eqn:Pdecomposition}
\end{equation}
in which $P_{s}$ is the density matrix corresponding to the system
part whose block corresponding to the bath part,
$(P_{s})_{\mc{I}_{b},\mc{I}_{b}}$, approximately vanishes. Similarly
$P_{0,b}$, called the \REV{bath} projector, is the density matrix from the bath part whose block
corresponding to the system part, $(P_{0,b})_{\mc{I}_{s},\mc{I}_{s}}$,
approximately vanishes. The decomposition of the system and
bath part is performed using projectors, thus leading
to the name of PET. Such a decomposition can be \textit{in-principle
exact}.
The subscript $0$ indicates that $P_{0,b}$ \REV{is computed from a \textit{reference}
system, thus only obtained approximately.}
Furthermore, $P_{s}$ is constrained by $P_{0,b}$
according to the
orthogonality condition
\begin{equation}
  P_{0,b} P_{s} = 0.
  \label{eqn:Porthogonal}
\end{equation}
The condition~\eqref{eqn:Porthogonal} acts as a soft ``boundary condition'' for a modified Kohn-Sham problem,
of which the number of eigenvectors to be computed can be much smaller than
$N_{e}$. Hence PET reduces the computational cost compared to
\REV{solving} \eqref{eqn:KSDFT} by reducing the number of eigenvectors and eigenvalues
to compute.

\noindent\textbf{Related works:}

From a practical perspective, PET can be seamlessly integrated into many
electronic structure software packages, given its alluring matrix-free
nature, \ie PET only requires matrix-vector multiplication operation of the form
$H\psi$. Thus, it can be applicable even
when $H$ is a dense matrix such as in the planewave discretization, or when
explicit access to $H$ is not readily available.

This is in contrast to, \eg the widely used Green's function embedding
methods (see, \eg
~\cite{BernholcLipariPantelides1978,ZellerDederichs1979,WilliamsFeibelmanLang1982,KellyCar1992,LiLinLu2018}),
\REV{where explicit access to $H$ is usually
required to compute Green's functions
of the form $G(z):=(z-H)^{-1}$. In addition, when a large basis set such as a planewave
discretization is used, even storing the Green's functions can be challenging.
However, when a small basis set is used, and $H$ is a sparse matrix,
Green's function embedding methods can be combined with fast
algorithms~\cite{LiLinLu2018} to yield a lower computational complexity than
that of PET. It may also perform better for systems with small
gaps.}

\noindent\textbf{Contribution:}

The contribution of this paper is two-fold:  First,  we provide a
mathematical understanding of PET from a linear algebra perspective,
which can be concisely stated as an energy minimization problem with \REV{extra
orthogonality constraints}. The corresponding Euler-Lagrange equation from the
energy minimization problem gives rise to a modified Kohn-Sham problem,
and the original PET formulation can be understood as a penalty method for
implementing the \REV{extra orthogonality constraint~\eqref{eqn:Porthogonal}}.
We then extend the formalism to the nonlinear case as
in KSDFT.

\REV{Second, we found that the standard perturbation analysis can not be applied
directly to PET. However, through a proper choice of the basis set, it is possible to
reformulate PET in a form suitable for such analysis, which allows us to compute a
perturbative correction. In addition, we show that such correction only contributes
to the bath part\footnote{\REV{We refer readers to the main text
(in particular, Section~\ref{sec:perturbationlinear})
for the formula of the aforementioned perturbation.}}.}

Our numerical results for real chemical systems confirm the
effectiveness of the method. In particular, we
find that the method can reach below the chemical accuracy (1 kcal/mol, or
$0.0016$ au) even
when applied to very small systems, such as the computation of the
ground state energy of a SiH$_3$F molecule from the reference of a
SiH$_4$ molecule. We also demonstrate the accuracy of the energy and the
atomic force for the PET and the perturbatively corrected PET using
other molecules such as benzene and anthracene.

\noindent\textbf{Organization:}

This paper is organized as follows. We derive PET for linear
problems in Section~\ref{sec:linear}, and introduce the first order
perturbative correction to PET in Section~\ref{sec:perturbationlinear}.
We then generalize the discussion to
nonlinear problems in Section~\ref{sec:nonlinear}. We discuss the
strategy to evaluate the \REV{bath} projector using localization methods in
Section~\ref{sec:implementation}.  We then present the numerical results
in Section~\ref{sec:numer}, followed by the conclusion and discussion in
Section~\ref{sec:conclusion}.

\section{PET for linear problems}\label{sec:linear}

We first introduce PET in the context of solving a linear eigenvalue problem.
Let $H \in \CC^{N\times N}$ be a Hermitian matrix, whose eigenvalues
are ordered non-decreasingly as $\lambda_1 \leq \lambda_2 \leq \ldots
\leq \lambda_{N_e} < \lambda_{N_e+1} \leq \ldots \leq \lambda_{N}$. Here
we assume that there is a positive energy gap $\Delta_{g} =
\lambda_{N_e+1}-\lambda_{N_{e}}$.

Consider the following energy minimization problem
\begin{equation}
  E = \underset{\scriptsize\begin{array}{c} P^2 = P, P^* = P \\ \Tr P = N_e \end{array} }{\inf} \cE [P],
  \label{eq:energy_linear}
\end{equation}
where $E$ is called the energy, and the energy functional $\cE[P]$ is defined as
\begin{equation}
  \cE[P] := \Tr [HP].
\end{equation}
Note that the condition $P=P^2$ requires $P$ to be a projector, with
eigenvalues being either $0$ or $1$. The trace condition
ensures that there are precisely $N_{e}$ eigenvalues that are equal to
$1$. Proposition~\ref{prop:linearks} states that the minimizer is attained by
solving a linear eigenvalue problem.  This is a well known result in
linear algebra; nonetheless, we provide its proof here in order to motivate the
derivation for PET later.

\begin{prop}\label{prop:linearks}
  \REV{Let $H \in \CC^{N\times N}$ be a Hermitian matrix and} assume that there is a positive gap between the $N_{e}$-th and
  $(N_{e}+1)$-th eigenvalue of $H$.
  Then the variational problem~\eqref{eq:energy_linear} has a unique
  minimizer, \REV{ denoted by $P$}, which is given by the solution to the following linear
  eigenvalue problem
  \begin{equation}
    H \Psi = \Psi \Lambda, \quad P=\Psi\Psi^{*}.
    \label{eqn:linear}
  \end{equation}
  Here $(\Psi,\Lambda)$ are the lowest $N_{e}$ eigenpairs of $H$.
\end{prop}

\begin{proof}
  Since $H$ is a Hermitian matrix, it can be diagonalized as
  \begin{equation}
    H = \hat{\Psi} \hat{\Lambda} \hat{\Psi}^{*}.
  \end{equation}
  Here $\hat{\Lambda}=\text{diag}[\lambda_{1},\ldots,\lambda_{N}]\in \RR^{N\times N} $ is a diagonal matrix
  containing all the eigenvalues of $H$ ordered non-decreasingly, and
  $\hat{\Psi}\in \CC^{N\times N}$ is a unitary matrix with its first
  $N_{e}$ columns given by $\Psi$.
  Then
  \begin{equation}
    \cE[P] = \Tr[HP] =  \Tr[\hat{\Psi} \hat{\Lambda} \hat{\Psi}^* P] = \Tr[\hat{\Lambda} \hat{\Psi}^*
    P \hat{\Psi}] = \Tr[\hat{\Lambda} \hat{P}] = \sum_{i =1}^{N}
    \lambda_i \hat{P}_{ii}=: \hat{\cE}[\hat{P}],
  \end{equation}
  where $\hat{P} = \hat{\Psi}^* P \hat{\Psi}$ is the density matrix
  with respect to the basis given by $\hat{\Psi}$.
  Thus, Eq.~\eqref{eq:energy_linear} is equivalent to
  \begin{equation}
    E = \underset{\scriptsize\begin{array}{c} \hat{P}^2 = \hat{P},
      \hat{P}^* = \hat{P} \\ \Tr \hat{P} = N_e \end{array} }{\inf} \hat{\cE} [\hat{P}].
      \label{eq:energy_rotated}
  \end{equation}
  Since $\lambda_{N_e+1}-\lambda_{N_{e}}>0$, the minimizer is achieved by setting
  \begin{equation}
    \hat{P}_{ii} = \left \{ \begin{array}{ll} 1, & \text{if } i \leq N_e, \\
      0, & \text{if } i > N_e. \end{array} \right.
  \end{equation}
  Finally, given that $\hat{P}$ is an idempotent matrix, we have
  that all its eigenvalues are either $0$ or $1$. \REV{Since
  $\hat{P}$ is a projection operator, its eigenvalues are
  bounded between $0$ and $1$. Since each diagonal entry $\hat{P}_{ii}$
  is already $1$ or $0$, $\hat{P}_{ii}$ is an eigenvalue. The
  corresponding eigenvector is $e_{i}$,  the $i$-th column of the
  identity matrix. }
  Thus $\hat{P}$ is a diagonal matrix, with ones and zeros at the main diagonal, \ie
  \begin{equation}
    \hat{P}_{ij} = \left \{ \begin{array}{ll} 1, & \text{if } i=j \text{ and } i \leq N_e, \\
      0, & \text{otherwise}. \end{array} \right.
  \end{equation}
  This is the unique minimizer.  Thus
  \begin{equation}
    P = \hat{\Psi} \hat{P} \hat{\Psi}^* = \Psi\Psi^{*}.
  \end{equation}
  is the unique minimizer of \eqref{eq:energy_linear}, where $\Psi$ is
  given by the first $N_{e}$ columns of $\hat{\Psi}$.
\end{proof}

%\subsection{PET formulation}\label{sec:petlinear}

When $N$ and $N_{e}$ are large, the solution of the linear eigenvalue
problem~\eqref{eqn:linear} can be expensive. However, if we have
already solved the eigenvalue problem for a reference matrix
$H_{0}$, and we would like to solve the eigenvalue problem for another
matrix $H$ \REV{such} that $H-H_{0}$ is approximately zero outside the matrix
block given by the index set $\mc{I}_{s}$. In such a case, PET aims at
reducing the computational cost by solving a modified eigenvalue problem
that involves a much smaller number of eigenvectors.

More specifically, for a reference system $H_0\in \CC^{N\times N}$, let
$P_0$ be the minimizer of the following problem
\begin{equation}
  E_0 = \underset{\begin{array}{c} P^2 = P, P^* = P \\ \Tr P
    = N_e^{0} \end{array} }{\inf} \Tr[H_0 P].
\end{equation}
We split the minimizer as
\begin{equation}
  P_0 = P_{0,b} + P_{0,s}.
  \label{eqn:P0ansatz}
\end{equation}
\REV{Here $P_{0,b}$ and $P_{0,s}$ are called system and bath projector, respectively,
and are projectors themselves, \ie}
\begin{equation}
  P_{0,b}^2=P_{0,b}, \quad  P_{0,s}^2=P_{0,s}.
  \label{eqn:proj_ref}
\end{equation}
The rank of $P_{0,b}$ is denoted by $N_{b}:=\Tr
P_{0,b}$.  \REV{Here we use the symbol $N_{b}$ instead of $N_{0,b}$ to
emphasize that the rank of the bath projector $P_{0,b}$ remains the same
before and after the perturbation, as will be seen in the discussion later.}
We assume that $N_{b}\approx N_{e}^{0}$, and hence the rank of
$P_{0,s}$ is much smaller than $N_{b}$.  The splitting
procedure~\eqref{eqn:P0ansatz} is by no means unique; we will discuss
one possible method based on localization techniques to choose
$\Psi_{0,b}$ in Section~\ref{sec:implementation}.

Together with $P_0^2 = P_0$, we have
\begin{equation}
  P_0^2 =  (P_{0,b} + P_{0,s})^2 = P_{0,b}^2 + P_{0,s}^2 +
  P_{0,b}P_{0,s} + P_{0,s}P_{0,b} = P_{0,b} + P_{0,s}.
  \label{}
\end{equation}
\REV{Using Eq.~\eqref{eqn:proj_ref}, we have $P_{0,b}P_{0,s} +
P_{0,s}P_{0,b} = 0$.  Then
\[
P_{0,b}P_{0,s}P_{0,s} + P_{0,s}P_{0,b} P_{0,s} =
P_{0,b}P_{0,s}(I+P_{0,s})=0.
\]
Since $I+P_{0,s}$ is invertible, we arrive at the orthogonality
condition $P_{0,b}P_{0,s} = 0$.
}
It is also convenient to write
\begin{equation}
  P_{0,b}=\Psi_{0,b}\Psi_{0,b}^{*}, \quad
  \Psi_{0,b}^{*}\Psi_{0,b}=I_{N_{b}}.
  \label{}
\end{equation}
By proper rotation\footnote{This can be achieved by solving the
eigenvalue problem $\left (\Psi_{0,b}^*H_0\Psi_{0,b}\right ) C_b = C_b
\Lambda_{0,b}$, and redefining $\Psi_{0,b}$ to be $\Psi_{0,b} C_b$.} of the
$\Psi_{0,b}$ matrix, without loss of generality we may assume that
\begin{equation}
  \Psi_{0,b}^{*}H_{0}\Psi_{0,b} := \Lambda_{0,b}
  \label{eqn:Psi0b}
\end{equation}
is a diagonal matrix.
  We define
$\mathcal{B}_{0}:=\text{span}\{\Psi_{0,b}\}$,
with its orthogonal complement denoted by $\mathcal{B}_{0}^{\perp}$.

The main ansatz in PET is that the density matrix $P$ can be split
as
\begin{equation}
  P = P_{0,b} + P_s, \label{eq:ansatz}
\end{equation}
where $P_{s}^2=P_{s}$ is also a projector, and $P_{0,b}$ is \REV{
a bath} projector as in~\eqref{eqn:P0ansatz}. Similar to the
discussion above, we arrive at the orthogonality
condition~\eqref{eqn:Porthogonal}. Since the rank of $P_{0,b}$ is
already $N_{b}$, the rank of $P_{s}$ is thus equal to
$N_{s}:=N_{e}-N_{b}$, and we expect that $N_{s}\ll N_{b}$. Note that the
dimension of $H_{0}$ and $H$ must be the same, but $N_{e}^{0}$ and
$N_{e}$ can be different. Thus the ranks of $P_{0,s}$ and $P_{s}$ can
also be different. This is necessary in the context of KSDFT, where the
system part can involve different numbers and/or types of atoms
from that \REV{in the reference system}.

With the \REV{bath} projector fixed PET, solves the following constrained
minimization problem only with respect to $P_{s}$:
\begin{equation}
  E^{\text{PET}} = \underset{\begin{array}{c} P_s^2 = P_s,
    P_s^* = P_s \\ P_{0,b}P_s = 0,  \Tr P_s = N_s \end{array} }{\inf}
    \Tr [ H(P_s + P_{0,b})]. \label{eq:energy_pet_linear}
\end{equation}
Compared to~\eqref{eq:energy_linear}, we find that PET restrains the
feasibility set of density matrices to those satisfying the
ansatz~\eqref{eq:ansatz}. Hence by the variational principle
$E^{\text{PET}} \ge E$ provides an upper bound of the energy.
Parallel to Proposition~\ref{prop:linearks}, the minimizer
of~\eqref{eq:energy_pet_linear} is uniquely obtained by a modified
linear eigenvalue problem. This is given in
Proposition~\ref{prop:linearpet}.

\begin{prop}[Projection based embedding]\label{prop:linearpet}
  Let $H\vert_{\mathcal{B}_{0}^{\perp}}$ be the restriction of
  $H$ to the subspace $\mathcal{B}_{0}^{\perp}$, and assume that there
  is a positive gap between the $N_{s}$-th and $(N_{s}+1)$-th eigenvalue
  of $H\vert_{\mathcal{B}_{0}^{\perp}}$.
  Then the variational problem~\eqref{eq:energy_pet_linear} has a unique
  minimizer, denoted by $P_{s}$, which is given by the solution to the following linear
  eigenvalue problem
  \begin{equation}
    H\vert_{\mathcal{B}_{0}^{\perp}}\Psi_{s} = \Psi_{s} \Lambda_{s}, \quad
    P_{s}=\Psi_{s}\Psi_{s}^{*}.
    \label{eqn:linear_pet}
  \end{equation}
  Here $(\Psi_{s},\Lambda_{s})$ are the lowest $N_{s}$ eigenpairs of
  $H\vert_{\mathcal{B}_{0}^{\perp}}$.
\end{prop}

\begin{proof}
First, the orthogonality condition $P_{0,b}P_s = 0$ implies that all
columns of $P_{s}$ should be in the subspace $\mathcal{B}_{0}^{\perp}$.
Using the relation
\begin{align*}
  \Tr [ (I - P_{0,b})H(I - P_{0,b}) P_s ] & = \Tr [ H P_s - H P_{0,b}P_s - P_{0,b} H P_s + P_{0,b}HP_{0,b}P_s ], \\
                                         & = \Tr [ H P_s - H P_{0,b}P_s - H P_s P_{0,b} + HP_{0,b}P_sP_{0,b} ], \\
                                         & = \Tr [ H P_s ],
\end{align*}
we find that \eqref{eq:energy_pet_linear} is equivalent to the following minimization problem:
\begin{equation}
  E^{\text{PET}} = \underset{\begin{array}{c} P_s^2 = P_s, P_s^* = P_s \\
    P_{0,b}P_s = 0,  \Tr P_s = N_s \end{array} }{\inf} \Tr [ (I -
    P_{0,b})H(I - P_{0,b}) P_s] +\Tr [H P_{0,b}].
\end{equation}
Given that $\Tr [H P_{0,b}]$ is a constant, we only need to focus on the
first term  $\Tr [(I - P_{0,b})H(I - P_{0,b}) P_s]$. The matrix $(I -
P_{0,b})H(I - P_{0,b})$, \REV{sometimes} called the Huzinaga operator in
the quantum chemistry literature~\cite{Huzinaga1971}, is
Hermitian and is identical to $H$ when restricted to the subspace
$\mathcal{B}_{0}^{\perp}$. With some abuse of notation,
$H\vert_{\mathcal{B}_{0}^{\perp}}$ can be diagonalized as
\begin{equation}
  H\vert_{\mathcal{B}_{0}^{\perp}} = \hat{\Psi} \hat{\Lambda} \hat{\Psi}^{*}.
  \label{eqn:Hperpdecomp}
\end{equation}
Since the dimension of $\mathcal{B}_{0}^{\perp}$ is $N-N_{b}$,
$\hat{\Lambda}=\text{diag}[\hat{\lambda}_{1},\ldots,\hat{\lambda}_{N-N_{b}}]\in \RR^{(N-N_{b})\times (N-N_{b})} $ is a diagonal matrix
containing all the eigenvalues of $H\vert_{\mathcal{B}_{0}^{\perp}}$ ordered non-decreasingly, and
$\hat{\Psi}\in \CC^{N\times (N-N_{b})}$ is given by orthogonal columns
of \REV{an} unitary matrix in the subspace $\mathcal{B}_{0}^{\perp}$. Then
\begin{align*}
  \Tr [ (I - P_{0,b})H(I - P_{0,b}) P_s ] & =
  \Tr[\hat{\Psi} \hat{\Lambda} \hat{\Psi}^* P_{s}] = \Tr[\hat{\Lambda} \hat{\Psi}^*
  P_{s} \hat{\Psi}] \\
  &= \Tr[\hat{\Lambda} \hat{P}_{s}] =
  \sum_{i=1}^{N-N_{b}}
  \hat{\lambda}_i \hat{P}_{ii}.
\end{align*}
Here $\hat{P}_{s}$ is the matrix representation of $P_{s}$ with respect
to the basis $\hat{\Psi}$ of the subspace $\mathcal{B}_{0}^{\perp}$.

Thus similar to the proof of Proposition~\ref{prop:linearks}, we
arrive at the following minimization problem
\begin{equation}
  \label{eqn:PET_minimize}
  E^{\text{PET}} = \underset{\begin{array}{c}
    \hat{P}_{s}^2 = \hat{P}_{s}, \hat{P}_{s}^* =
    \hat{P}_{s} \\ \Tr \hat{P}_{s} = N_s \end{array} }{\inf} \Tr
    [ \hat{\Lambda} \hat{P}_{s}] +\Tr [H P_{0,b}],
\end{equation}
whose minimizer is given by
\begin{equation}
    (\hat{P}_{s})_{ij} = \left \{ \begin{array}{ll} 1, & \text{if } i=j
      \text{ and } i \leq N_s, \\ 0, & \text{otherwise}, \end{array}
      \right.
\end{equation}
and
\begin{equation}
  P_s = \hat{\Psi} \hat{P}_{s} \hat{\Psi}^* =
  \Psi_{s}\Psi_{s}^{*}.
\end{equation}
Here $\Psi_{s}$ are the first $N_{s}$ columns of $\hat{\Psi}$
corresponding to the lowest $N_{s}$ eigenvalues.
\end{proof}

\REV{We note that, even if $H_{0}$ and $H$ have a positive energy gap, it
may not be necessarily the case for  $H\vert_{\mathcal{B}_{0}^{\perp}}$. Therefore we need
to explicitly make this assumption in the proposition.}

\REV{Furthermore, we point out that $\{\hat{\lambda}_{i}\}$, the eigenvalues of
$H\vert_{\mathcal{B}_{0}^{\perp}}$, are not, in general, a subset of $\{\lambda_{i}\}$, the eigenvalues of $H$. Nonetheless,
according to Eq.~\eqref{eqn:PET_minimize}, $E^{\text{PET}}$ can be
computed in terms of the trace
\[
E^{\text{PET}} = \Tr[ H(P_s + P_{0,b})]=\sum_{i=1}^{N_{s}} \hat{\lambda}_{i}+\Tr [H P_{0,b}],
\]
which yields an upper bound to the energy $E$.}

\REV{In addition, when computing $\Psi_{s}$, all the vectors $\Psi_{0,b}$ lie
in the null space of $(I - P_{0,b})H(I - P_{0,b})$ which do not belong to the
range of $H\vert_{\mathcal{B}_{0}^{\perp}}$, thus they should be avoided
in the computation}.  This issue becomes noticeable when
$\hat{\lambda}_{N_{s}}>0$, and it would be incorrect to simply select
the first $N_{s}$ eigenpairs of $(I - P_{0,b})H(I - P_{0,b})$.  One
practical way to get around this problem is to add a negative shift
$c$, so that all the first $N_{s}$ eigenvalues of the matrix $(I -
P_{0,b})(H+cI)(I - P_{0,b})$ become negative.

This issue can also be automatically taken care of
by applying the projector $I - P_{0,b}$ to the computed eigenvectors
in an iterative solver, so that the computation is restricted to the
subspace of interest $\mathcal{B}_{0}^{\perp}$.

\begin{rem}
  The original formulation of PET~\cite{ManbyStellaGoodpasterEtAl2012} can be understood as a
  penalty formulation to implement the orthogonality constraint, \ie
  \begin{equation}
    E^{\text{PET},\mu} = \underset{\begin{array}{c} P_s^2 =
      P_s, P_s^* = P_s \\
      \Tr P_s = N_s \end{array} }{\inf} \Tr [ H(P_s
      + P_{0,b})] + \mu \Tr[P_{0,b}P_s]. \label{eq:energy_pet_penalty}
  \end{equation}
  This advantage of the penalty formulation is that the domain of
  $P_{s}$ has the same form as that in Proposition~\ref{prop:linearks}
  but with a modified energy functional.  The corresponding
  Euler-Lagrange equation is given by the eigenvalue problem for the
  matrix $H + \mu P_{0,b}$, and $P_{s}$ is the density matrix
  corresponding to the first $N_{s}$ eigenpairs. Therefore by selecting
  the penalty $\mu$ to be sufficiently large (in practice it is set to
  $10^{6}$ or larger),  the orthogonality condition is approximately
  enforced.
\end{rem}

\section{Perturbative correction to PET for linear problems}\label{sec:perturbationlinear}
\REV{
In the following we define,
\begin{equation}
  \delta H := H - H_0,
\end{equation}
and the PET projector,
\begin{equation}\label{eqn:P_PET}
  P^{\text{PET}} = P_s+P_{0,b},
\end{equation}
where $P_s$ is given by the solution of \eqref{eqn:linear_pet}.}
\subsection{Consistency}
First, we would like to verify that PET is a consistent theory: when $H=H_{0}$
and $N_{e}=N_{e}^{0}$, for any choice of the bath projector $P_{0,b}$, the
minimizers from~\eqref{eq:energy_linear}
and~\eqref{eq:energy_pet_linear} should yield the same density matrix.
This is ensured by Proposition~\ref{prop:consistentpet}.

\begin{prop}[Consistency of PET]\label{prop:consistentpet}
  When $H = H_0$ and $N_{e}=N_{e}^{0}$, the solution to PET
  satisfies $P_{s}=P_{0,s}$.
\end{prop}
\begin{proof}
  By the Courant-Fischer \REV{min-max} theorem,
  \[
  \hat{\lambda}_{N_{s}+1} =
  \max_{\stackrel{\text{dim}(S)=N-N_{e}}{S\subset
  \mathcal{B}_{0}^{\perp}}} \min_{u\in
  S\backslash\{0\}} \frac{u^{*}
  H\vert_{\mathcal{B}_{0}^{\perp}}u}{u^{*}u} \le
  \max_{\text{dim}(S)=N-N_{e}} \min_{u\in S\backslash\{0\}} \frac{u^{*}
  Hu}{u^{*}u} = \lambda_{N_{e}+1}.
  \]
  \REV{Furthermore}, when $H=H_{0}$, all eigenvectors of $H$ corresponding to eigenvalues
  $\lambda_{N_{e}+1}$ and above are in the subspace
  $\mathcal{B}_{0}^{\perp}$, and hence $\hat{\lambda}_{N_{s}+1}\ge
  \lambda_{N_{e}+1}$. Therefore
  \[
  \hat{\lambda}_{N_{s}+1}=\lambda_{N_{e}+1}.
  \]
  Again using the Courant-Fischer \REV{min-max} theorem we have
  \begin{equation}
    \hat{\lambda}_{N_{s}}\le \lambda_{N_{e}}.
  \end{equation}
  Hence the gap condition of $H$, \ie $\lambda_{N_{e}+1}-\lambda_{N_{e}}>0$, implies
  that the gap condition for Proposition~\ref{prop:linearpet} holds, \ie
  \[
  \hat{\lambda}_{N_{s}+1}-\hat{\lambda}_{N_{s}}>0.
  \]

  Since the minimizer of PET is obtained from a constrained domain of
  the density matrix, we have
  \begin{equation}
    E = \underset{\scriptsize\begin{array}{c} P^2 = P, P^* = P \\ \Tr P = N_e \end{array} }{\inf} \Tr [H_0 P] \leq \underset{\scriptsize\begin{array}{c} P_s^2 = P_s, P_s^* = P_s \\ P_{0.b}P_s^* = 0, \Tr P_s = N_s \end{array} }{\inf} \Tr [ H_0(P_s + P_{0,b})].
    \end{equation}
  where $P_s = P_{0,s}$ already achieves the minimum. By the uniqueness
  of the minimizer in Proposition~\ref{prop:linearpet}, we have $P_s = P_{0,s}$.
\end{proof}

\begin{rem}
  The proof of Proposition~\ref{prop:consistentpet} is not entirely
  straightforward. This is mainly due to the fact that $P_{0,b}$ is
  obtained through some linear combination of eigenvectors of
  $H_{0}$ corresponding to the lowest $N_{e}$ eigenvalues. Hence
  $H$ and $P_{0,b}$ generally do not commute even when $H=H_{0}$.
  Nonetheless, the consistency of PET implies that PET is an in
  principle exact theory, given the proper choice of the reference
  projector $P_{0,b}$.
\end{rem}

\REV{
\begin{rem}\label{rem:dm_energy_pet}
  Assume $N_e = N_e^0$, then we have that
\begin{align*}
  \norm{P-P^{\text{PET}}} & \leq \norm{ P - P_0 + P_0 - P^{\text{PET}}_0 + P_0^{\text{PET}} - P^{\text{PET}} },\\
                          & \leq \norm{ P - P_0 } + \norm{P_0 - P^{\text{PET}}_0} + \norm{ P_0^{\text{PET}}- P^{\text{PET}} },\\
                          & \leq \norm{ P - P_0 }  + \norm{ P_0^{\text{PET}}- P^{\text{PET}} }.
\end{align*}
Here we used $P_0 = P^{\text{PET}}_0$ by Proposition ~\ref{prop:consistentpet} and $\norm{\cdot}$ means the operator norm.
In addition, given that $H_0$ has a positive energy gap, and $\norm{\delta H}$ is sufficiently small we have, by continuity, that
\[
\norm{ P - P_0 } \sim \mathcal{O}(\norm{\delta H}).
\]
By the proof of Proposition ~\ref{prop:consistentpet}, $H\vert_{\mathcal{B}_{0}^{\perp}}$ also has a positive gap.
Using \eqref{eqn:P_PET} we have that
\[
\norm{ P_0^{\text{PET}}- P^{\text{PET}} }  = \norm{ P_s - P_{0,s }} \sim \mathcal{O}(\norm{\delta H}),
\]
thus resulting in
\begin{equation}
  \norm{P-P^{\text{PET}}}  \sim \mathcal{O}(\norm{\delta H}).
\end{equation}
  In addition, Proposition~\ref{prop:consistentpet} states that when $\delta H=0$,
  there is zero-th order consistency for the energy $E=E^{\text{PET}}$. The we can use a standard perturbative
  argument coupled with the computation above, and the fact that $P^{PET}$ lies in the
  feasible set to obtain
  \[
  \abs{E-E^{\text{PET}}} \sim \Or(\norm{\delta H}^2).
  \]
\end{rem}
}

\subsection{Perturbation}
\REV{In the following discussion, we derive a perturbative correction
to the density matrix when $H\approx H_{0}$.
Unfortunately, standard perturbation analysis for eigenvalue problems
do not apply directly given that PET depends on the solution of two separate eigenvalue
problems: one from $H_{0}$, to determine the projector
$P_{0,b}$; and other from
$H\vert_{\mathcal{B}_{0}^{\perp}}$, to compute $P_s$. }

% In order to bypass this difficulty, \REV{we show that
% PET can also be viewed as solving one} eigenvalue problem but in a rotated
% basis, and we use standard perturbative analysis on the rotated problem to
% compute the \REV{perturbation}. The results are then rotated back to the
% original basis.

In order to bypass this difficulty, \REV{though a proper choice of the basis set,
the two eigenvalue problems can be formally combined into one. In this rotated basis,
we use standard perturbative analysis to compute a perturbative correction. The result
are then rotated back to the original basis. }

Let us split the set of vectors $\hat{\Psi}$ from the
eigen-decomposition~\eqref{eqn:Hperpdecomp} as
\[
\hat{\Psi} = [\Psi_{s},\Psi_{u}],
\]
where $\Psi_{s}$ corresponds the projector $P_{s}$ according to
Proposition~\ref{prop:linearpet}, \REV{and $\Psi_{u}$ denotes the rest of
the vectors. Here the subscript $u$ stands for {\it unoccupied orbitals}
following the terminology of KSDFT}. Correspondingly the diagonal matrix
$\hat{\Lambda}$ is split into the block diagonal form as
\[
\hat{\Lambda} = \begin{bmatrix}
  \Lambda_{s} & 0\\
  0 & \Lambda_{u}
  \end{bmatrix}.
\]
\REV{We combine $\hat{\Psi}$ and $\Psi_{0,b}$ from \eqref{eqn:Psi0b}}, to form a unitary $N\times N$ matrix
\begin{equation}
  W := [ \Psi_{0,b}, \Psi_{s},\Psi_{u}],
  \label{eqn:Wbasis}
\end{equation}
and the matrix representation of $H$ with respect to the basis
$W$, denoted by $H_{W}$, can be written as
\begin{align*}
  H_{W} = W^{*}H W =&
  \begin{bmatrix}
    \Psi_{0,b}^{*} H\Psi_{0,b} & \Psi_{0,b}^{*} H\Psi_{s} & \Psi_{0,b}^{*}H\Psi_{u} \\
    \Psi_{s}^{*} H\Psi_{0,b}   & \Psi_{s}^{*}   H\Psi_{s} & \Psi_{s}^{*}  H\Psi_{u}\\
    \Psi_{u}^{*} H\Psi_{0,b}   & \Psi_{u}^{*}   H\Psi_{s} & \Psi_{u}^{*}  H\Psi_{u}\\
  \end{bmatrix}, \\
  =& \begin{bmatrix}
    \Psi_{0,b}^{*} H\Psi_{0,b} &  \Psi_{0,b}^{*} H\Psi_{s} &   \Psi_{0,b}^{*}H\Psi_{u}\\
    \Psi_{s}^{*} H\Psi_{0,b}   &  \Lambda_{s}              &  0 \\
    \Psi_{u}^{*}H\Psi_{0,b}    &  0                        &  \Lambda_{u}
  \end{bmatrix}.
\end{align*}
In the second equality, we have used the fact that all columns of
$\Psi_{s},\Psi_{u}$ belong to $\mathcal{B}_{0}^{\perp}$, and consist of
eigenvectors of $H\vert_{\mathcal{B}_{0}^{\perp}}$ associated with
different sets of eigenvalues. Hence the
inner product $\Psi_{s}^{*}H\Psi_{u}$ vanishes.
Again we note that not all off-diagonal matrix blocks vanish even when $H=H_{0}$.

From this perspective, we find that PET makes two approximations: first, it discards the off-diagonal matrix blocks, so that
\begin{align*}
  H_{W} \approx \begin{bmatrix}
    \Psi_{0,b}^{*} H\Psi_{0,b} &  0           &   0\\
    0                          &  \Lambda_{s} &  0 \\
    0                          &  0           &  \Lambda_{u}
  \end{bmatrix}, \\
\end{align*}
and second, it replaces the block corresponding to the interactions within 
the bath by \REV{$\Psi_{0,b}^{*} H_0 \Psi_{0,b}$ which is equal to $\Lambda_{0,b}$ 
following \eqref{eqn:Psi0b}}. The resulting matrix,
\[
H_{W}^{\text{PET}} =
\begin{bmatrix}
  \Lambda_{0,b} & 0 & 0\\
  0 & \Lambda_{s} & 0\\
  0 & 0 & \Lambda_{u}
\end{bmatrix},
\]
is already diagonalized in the $W$-basis.  Assume that the first
$N_{e}$ eigenvalues of $H_{W}^{\text{PET}}$ include all the diagonal
entries of $\Lambda_{0,b}$ (according to
Proposition~\ref{prop:consistentpet}, this is at least valid when
$H=H_{0}$ and $N_{e}=N_{e}^{0}$), we find that the density matrix in the
$W$-basis takes the block diagonal form
\[
P^{\text{PET}}_{W} =
\begin{bmatrix}
  I_{N_{b}} & 0 & 0\\
  0 & I_{N_{s}} & 0 \\
  0 & 0 & 0\\
\end{bmatrix}.
\]
When rotated back to the standard basis, the density matrix becomes
\[
P^{\text{PET}} = W P^{\text{PET}}_{W} W^{*} = \Psi_{0,b}\Psi_{0,b}^{*}+
\Psi_{s}\Psi_{s}^{*} = P_{s} + P_{0,b},
\]
which is the PET solution.

One advantage of the representation in the $W$-basis is that the density
matrix $P^{\text{PET}}_{W}$ can be concisely written using the Cauchy
contour integral formula as
\begin{equation}
 P^{\text{PET}}_{W} = \frac{1}{2 \pi \I}\oint_{\mathcal{C}} (z I -
 H_{W}^{\text{PET}} )^{-1} \ud z.
  \label{}
\end{equation}
Here $\mathcal{C}$ is a contour in the complex plane surrounding only the
lowest $N_{e}$ eigenvalues of $H_{W}^{\text{PET}}$.

The first order perturbative correction to PET is then given by the neglected
off-diagonal matrix blocks $\Psi_{s}^{*} H\Psi_{0,b}$ and
$\Psi_{u}^{*}H\Psi_{0,b}$, and the diagonal term involving $\Psi_{0,b}^{*} (H-H_0)\Psi_{0,b}$.  The formula for the first order
perturbation is given in
Proposition~\ref{prop:perturblinearpet}.

\begin{prop}[First order perturbation]\label{prop:perturblinearpet}
  The first order perturbation to the density matrix from PET is
  given by
  \begin{equation}
    \delta P = \delta \Psi_{0,b} \Psi_{0,b}^{*} + \text{h.c.}
    \label{eqn:Pcorrection}
  \end{equation}
  where $\delta \Psi_{0,b}\in\CC^{N\times N_{b}}$ satisfies the equation
  \begin{equation}
   Q\left(\lambda_{i;0,b}I-H\right) Q\delta \psi_{i;0,b} =  Q (H
   \psi_{i;0,b}),
   \quad Q\delta \psi_{i;0,b} = \delta\psi_{i;0,b}.
    \label{eqn:firstorder}
  \end{equation}
  Here the projector $Q=I-(P_{s}+P_{0,b})=\Psi_{u}\Psi_{u}^{*}$.
  $\lambda_{i;0,b}$ is the $i$-th diagonal element of
  $\Lambda_{0,b}$, and $\psi_{i;0,b},\delta\psi_{i;0,b}$ are the
  $i$-th column of $\Psi_{0,b},\delta\Psi_{0,b}$, respectively.
  $\text{h.c.}$ stands for the \REV{Hermitian} conjugate of the first term.
\end{prop}
\begin{rem}
  Following the previous notation we may define the subspace
  $\mathcal{B}:=\text{span}\{\Psi_{s},\Psi_{0,b}\}$, and $Q$ is the
  projector on the orthogonal complement subspace
  $\mathcal{B}^{\perp}$.  Since $\lambda_{0,b}$ is separated from the
  spectrum of $H\vert_{\mathcal{B}_{0}^{\perp}}$,
  Eq.~\eqref{eqn:firstorder} has a unique solution in
  $\mathcal{B}^{\perp}$. Eq.~\eqref{eqn:Pcorrection} suggests that the
  first \REV{order} correction to the system part $P_{s}$ vanishes, and the
  correction only comes from the bath part $P_{0,b}$. Furthermore, the
  correction is traceless due to the condition $\Psi_{0,b}^{*}\delta
  \Psi_{0,b} = 0$.  This means that the density matrix after the first
  order correction preserves the trace
  of the projector, which is $N_{e}$. In the
  context of KSDFT, this means that the first order correction preserves
  the number of electrons in the system.
\end{rem}

\begin{rem}
  From Eq.~\eqref{eqn:firstorder} it may appear that the correction does
  not vanish even when $H=H_{0}$. However, note that
  $H_{0}\psi_{i;0,b}\in
  \mathcal{B}:=\text{span}\{\Psi_{s},\Psi_{0,b}\}$, we have
  $QH_{0}\psi_{i;0,b} = 0$, and hence the first order correction \REV{indeed}
  vanishes. This is consistent with
  Proposition~\ref{prop:consistentpet}.
\end{rem}

\begin{rem}\label{rem:effectperturb}
  The perturbative correction requires the solution of $N_{b}$ linear
  equations to correct the projector from the bath.  It seems that this
  diminishes the purpose of PET which reduces the number of eigenpairs
  to be computed from $N_{e}$ to $N_{s}$ from a practical perspective.
  Hence the advantage of the perturbative correction becomes more
  apparent in the nonlinear setup in \REV{Section}~\ref{sec:nonlinear}, where
  the perturbation only needs to be applied once after the
  self-consistency is achieved.
\end{rem}

\begin{rem}\label{rem:SternheimerEq}
\REV{
  We point out that Eq.~\eqref{eqn:firstorder} shares some similarities
  to the Sternheimer equation used in density functional perturbation
  theory~\cite{BaroniGironcoliDalEtAl2001}. The perturbative correction
  lies in the  subspace orthogonal to the range of $P_s+P_{0,b}$.
  However, unlike the Sternheimer equation, in our case
  $\lambda_{i,0,b}$ is not necessarily an eigenvalue of $H$.
}
\end{rem}

\begin{proof}
  Our strategy is to derive the first order perturbation in the
  $W$-basis, denoted by $\delta P_{W}$, and then obtain $\delta P$
  according to $\delta P=W\delta P_{W} W^{*}$.

  Let us first denote by
  \[
  \delta H^{\text{PET}}_{W}
  = \begin{bmatrix}
    \Psi_{0,b}^{*} \delta H \Psi_{0,b} & \Psi_{0,b}^{*} H\Psi_{s} &  \Psi_{0,b}^{*}H\Psi_{u}\\
    \Psi_{s}^{*} H\Psi_{0,b} & 0                        &  0\\
    \Psi_{u}^{*}H\Psi_{0,b}  & 0                        &  0
  \end{bmatrix}
  \]
  the neglected off-diagonal matrix blocks in PET. $\delta
  H^{\text{PET}}_{W}$ may not be small even when $H=H_{0}$, but its contribution to
  the density matrix must vanish according to
  Proposition~\ref{prop:consistentpet}, and hence can be \REV{formally} treated
  perturbatively.

  Let $P_{W} = P^{\text{PET}}_W + \delta P_{W}$, setting $G_{W}(z) = (z
  - H_W)^{-1}$ and $G(z)^{\text{PET}}_W = (z - H_W^{\text{PET}})^{-1}$,
  we have the Dyson equation
  \[ G_{W}(z) = G^{\text{PET}}_{W}(z) + G^{\text{PET}}_{W}(z)  \delta H^{\text{PET}}_{W} G_{W}(z). \]
  Thus using the Cauchy integral formulation we have that
\begin{align*}
  \delta P_{W} = & P_W - P^{\text{PET}}_W, \\
               = & \frac{1}{2 \pi \I} \oint_{\mathcal{C}} G_W(z) - G_W^{\text{PET}}(z) \ud z,\\
               = & \frac{1}{2 \pi \I} \oint_{\mathcal{C}} G_W^{\text{PET}}(z)  \delta H^{\text{PET}}_{W} G_W(z) \ud z,\\
               = & \frac{1}{2 \pi \I} \oint_{\mathcal{C}} (z I - H_{W}^{\text{PET}})^{-1} \delta H^{\text{PET}}_{W} (z I - H_{W})^{-1} \ud z.
\end{align*}
  By setting $H_{W} \approx H_{W}^{\text{PET}}$, the first order correction is
  \[
  \delta P_{W} = \frac{1}{2 \pi \I}\oint_{\mathcal{C}}
  (z I - H_{W}^{\text{PET}})^{-1} \delta H^{\text{PET}}_{W} (z I -
  H_{W}^{\text{PET}})^{-1} \ud z.
  \]
  Since $H_{W}^{\text{PET}}$ is a diagonal matrix, $\delta P_{W}$ should
  have the same matrix sparsity pattern as $\delta H^{\text{PET}}_{W}$, \ie
  \[
  \delta P_{W} = \begin{bmatrix}
    (\delta P_{W})_{b,b}  & (\delta P_{W})_{b,s}  &   (\delta P_{W})_{b,u}\\
    (\delta P_{W})_{s,b}  & 0                     &  0\\
    (\delta P_{W})_{u,b}  & 0                     &  0
  \end{bmatrix}.
  \]
  First we compute
  \[
  (\delta P_{W})_{b,s} = \frac{1}{2 \pi \I}\oint_{\mathcal{C}}
    (z I - \Lambda_{0,b})^{-1} \Psi_{0,b}^{*}H\Psi_{s}(z I -
    \Lambda_{s})^{-1} \ud z.
  \]
  Note that for any diagonal elements $\lambda_{i;0,b},\lambda_{j;s}$
  from $\Lambda_{0,b},\Lambda_{s}$, respectively, they are both enclosed
  in the contour $\mathcal{C}$.

  \REV{On the one hand}, if $\lambda_{i;0,b}\neq \lambda_{j;s}$,
  then
  \begin{align*}
  \frac{1}{2 \pi \I}\oint_{\mathcal{C}}
  (z  - \lambda_{i;0,b})^{-1} (z  - \lambda_{j;s})^{-1} \ud z
  & =   \frac{1}{2 \pi \I}\oint_{\mathcal{C}}
  \frac{ (z  -
  \lambda_{i;0,b})^{-1}-(z  - \lambda_{j;s})^{-1} }{\lambda_{i;0,b}-\lambda_{j;s}} \ud z \\
  & =
  \frac{1-1}{\lambda_{i;0,b}-\lambda_{j;s}} =  0.
  \end{align*}
  On the other hand, if $\lambda_{i;0,b}=\lambda_{j;s}$, \REV{then} 
  we would obtain an integral of the form
 \begin{equation*}
  \frac{1}{2 \pi \I}\oint_{\mathcal{C}}
  (z  - \lambda_{i;0,b})^{-2} \ud z,
 \end{equation*}
  which vanishes since the residue for the integrand is zero.

  For the term
  \[
  (\delta P_{W})_{b,b} = \frac{1}{2 \pi \I}\oint_{\mathcal{C}}
    (z I - \Lambda_{0,b})^{-1} \Psi_{0,b}^{*} \delta H \Psi_{0,b}(z I -
    \Lambda_{0,b})^{-1} \ud z,
  \]
  an analogous argument can be used to show that it vanishes.

  This means that the matrix blocks $(\delta P_{W})_{b,s}$, 
  $(\delta P_{W})_{s,b}$, and $(\delta P_{W})_{b,b}$ vanish, and
  the only nonzero matrix blocks are $(\delta P_{W})_{u,b}$ and its
  conjugate. Moreover,
  \begin{align*}
    (\delta P_{W})_{u,b} =& \frac{1}{2 \pi \I}\oint_{\mathcal{C}}
    (z I - \Lambda_{u})^{-1} \Psi_{u}^{*}H\Psi_{0,b}(z I -
    \Lambda_{0,b})^{-1} \ud z, \\
    = & \sum_{i} (\lambda_{i;0,b}-\Lambda_{u})^{-1}
    \Psi_{u}^{*} H \psi_{i;0,b}.
  \end{align*}
  Back to the standard basis
  \begin{align*}
    \delta P = & \Psi_{u}\sum_{i} (\lambda_{i;0,b}-\Lambda_{u})^{-1}
    \Psi_{u}^{*} H \psi_{i;0,b} \psi_{i;0,b}^{*} + \text{h.c.} \\
    = & \left(\sum_{i} \delta \psi_{i;0,b} \psi_{i;0,b}^{*}\right) + \text{h.c.}
  \end{align*}
  Here
  \[
  \delta \psi_{i;0,b} = \Psi_{u}(\lambda_{i;0,b}-\Lambda_{u})^{-1}
  \Psi_{u}^{*} H \psi_{i;0,b}.
  \]
  Using the projector $Q=\Psi_{u}\Psi_{u}^{*}$,  we find that $\psi_{i;0,b}$
  satisfies \eqref{eqn:firstorder}, and we prove the
  proposition.
\end{proof}

We summarize the perturbatively corrected PET in Algorithm~\ref{alg:perturbation}.

\begin{algorithm}
  \caption{Perturbatively corrected projection based embedding theory
  for linear eigenvalue problems.}
  \label{alg:perturbation}
\begin{algorithmic}[1]
  \Statex Input: $H$, $H_0$, $\Psi_{0,b}$, $\Psi_s$.
  \Statex Output: $\delta \Psi_{0,b}, \delta P$.
  \State Compute diagonal matrix  $\Lambda_{0,b} = \Psi_{0,b}^* H_0 \Psi_{0,b}$.
  \State Obtain the PET density matrix  $P^{\text{PET}} =
  \Psi_{0,b}\Psi_{0,b}^* +\Psi_s \Psi_s^*$.
  \State Compute the right-hand side $R = (I -P^{\text{PET}}) H \Psi_{0,b}$.
  \State Compute $\delta \Psi_{0,b}$ by solving \eqref{eqn:firstorder}.
  \State Obtain the perturbation to the density matrix
  $\delta P = \delta \Psi_{0,b} \Psi_{0,b}^* + \Psi_{0,b} \delta
  \Psi_{0,b}^* $.
\end{algorithmic}
\end{algorithm}

\subsection{Comparison of PET with the Rayleigh-Schr\"odinger Perturbation Theory}

\REV{Let us now have a more detailed comparison between the results from
Proposition~\ref{prop:perturblinearpet} with those from
the standard Rayleigh-Schr\"odinger (RS) perturbation theory. For simplicity we
consider the computation of the lowest, non-degenerate eigenpair $(\lambda,\psi)$
corresponding to $H$. The perturbation is computed with respect to the
lowest, non-degenerate eigenpair
$(\lambda_{0},\psi_{0})$ corresponding to the reference matrix $H_{0}$. We have
\[
\lambda_{0} = \psi_{0}^{*} H_{0} \psi_{0} = \Tr[H_{0} P_{0}],
\]
where the reference density matrix is $P_0=\psi_{0}\psi_{0}^{*}$.
The first order correction to the
eigenvalue is
\[
\delta \lambda^{(1)} = \psi_{0}^{*} \delta H \psi_{0} :=
\Tr[\delta H P_0],
\]
and the first order correction to the lowest eigenfunction can be computed as
\begin{equation}
\label{eqn:deltaPsi_RS}
  \delta \psi^{(1)} = Q_{0} (\lambda_{0}-H_{0})^{-1}Q_{0} (\delta H
\psi_{0}),
\end{equation}
where $Q_{0}=I-P_0$ projects to the subspace orthogonal to the range of
$P_{0}$.
This also gives the first order correction of the density matrix as
\[
\delta P^{(1)} = \delta \psi^{(1)}\psi_{0}^{*} + \psi_{0}
(\delta \psi^{(1)})^{*}.
\]
From \eqref{eqn:deltaPsi_RS} we have that $\delta \psi^{(1)}$ is
orthogonal to $\psi_{0}$, thus we can write
\[
\delta P^{(1)} P_{0} = \delta \psi^{(1)} \psi_{0}^{*}.
\]
In addition, the first order correction of the eigenfunction allows
us to compute the second order correction
to the eigenvalue as
\[
    \delta \lambda^{(2)} =\psi_{0}^{*} \delta H \delta \psi^{(1)} =
    \Tr[P_{0} \delta H \delta P^{(1)}].
\]
Let us then define
\[
P^{(1)} := P_{0} + \delta P^{(1)}, \quad \lambda^{(1)} := \lambda_{0} + \delta \lambda^{(1)} = \Tr[H P_{0}],
\]
and
\begin{equation}
  \begin{split}
    \lambda^{(2)}:= & \lambda_{0} + \delta \lambda^{(1)} + \delta \lambda^{(2)},  \\
    = &  \Tr[P_{0} H] + \Tr[P_{0} \delta H \delta P^{(1)}], \\
    = & \Tr[P_{0} H P_{0}] + \Tr[P_{0} H \delta P^{(1)}], \\
    = & \Tr[P_{0} H P^{(1)}].
  \end{split}
  \label{eqn:RS_E_order2}
\end{equation}
Here we have used $\Tr[P_{0} H_0 \delta P^{(1)}]=0$.  }

\REV{To summarize, the RS perturbation theory states that:
\begin{equation}
  \abs{\lambda-\lambda^{(1)}} \sim \Or(\norm{\delta H}^2),
  \,\, \norm{P-P^{(1)}} \sim \Or(\norm{\delta H}^2),\text{ and} \,\,
  \abs{\lambda-\lambda^{(2)}} \sim \Or(\norm{\delta H}^3).
  \label{}
\end{equation}
It is worth remarking that $\norm{P-P^{(1)}} \sim \Or(\norm{\delta
H}^2)$ does not imply $\abs{\lambda-\lambda^{(2)}} \sim \Or(\norm{\delta
H}^4)$. This is because the perturbed density matrix
$P^{(1)}$ satisfies the symmetry and trace condition, but not the
idempotency condition as in the feasible set of
the optimization problem~\eqref{eq:energy_linear}. Therefore, the
standard squared relation between the error of the eigenvalue and
the error of the eigenfunction does not hold.
In fact, Eq.~\eqref{eqn:RS_E_order2} suggests that the eigenvalue
computed to the correct order is not equal to $\Tr[H P^{(1)}]$, but
$\Tr[P_{0} H P^{(1)}]$.}

\REV{Motivated from Eq.~\eqref{eqn:RS_E_order2}, we may define the perturbed
energy in the PET formulation as
\begin{equation}\label{eqn:energy_pert}
    E^{\text{pert}} := \Tr [P^{\text{PET}} H P^{\text{pert}}],
\end{equation}
where
\begin{equation}
  P^{\text{pert}} := P^{\text{PET}} + \delta P.
\end{equation}
However, the perturbation theory used in Proposition~\ref{prop:perturblinearpet}
differs form the RS perturbation theory, in the sense that the
perturbation is performed with respect to $\delta H_W^{\text{PET}} = H_W - H_W^{\text{PET}}$,
rather than $\delta H$. In particular, $\delta H_W^{\text{PET}}$ may
not vanish even when $H=H_{0}$,  unless
\begin{equation}
  \Psi_{0,b}^{*} H_{0}\Psi_{0,s} = 0.
  \label{eqn:special_case_condition}
\end{equation}
In particular, Eq.~\eqref{eqn:special_case_condition} will be satisfied if the columns of $\Psi_{0,b}$ are eigenvectors of
$H_{0}$.  In such a case, the results of the perturbation theory of PET agree with those from the RS perturbation theory:
\[
  \norm{P-P^{\text{pert}}} \sim \Or(\norm{\delta H}^2), \quad \abs{E-E^{\text{pert}}} \sim \Or(\norm{\delta H}^3).
\]
This will be confirmed by the numerical results.
}

%Note that this does not prevent PET
%from being an exact theory when $H=H_{0}$ according to
%Proposition~\ref{prop:consistentpet}.  Furthermore,

\REV{However, when Eq.~\eqref{eqn:special_case_condition} is violated, $\norm{\delta H^{\text{PET}}_{W}}$ may not be small even when $\norm{\delta H}$ is small, and the perturbation theory
developed in Proposition~\ref{prop:perturblinearpet} holds only formally.
In such a case, the perturbation theory of PET does not improve the asymptotic convergence rate, and we have
\[
  \norm{P-P^{\text{pert}}} \sim \Or(\norm{\delta H}), \quad \abs{E-E^{\text{pert}}} \sim \Or(\norm{\delta H}^2).
\]
Interestingly, our numerical results indicate that even when the perturbative
correction is formal, the preconstant can be much reduced after the perturbation correction.}

\section{PET for nonlinear problems}\label{sec:nonlinear}

In this section we generalize PET and the perturbative expansion to
the nonlinear case as in KSDFT. First, define the energy functional
\begin{equation}
  \cE[P] = \Tr [ P H_{L} ] + E_{\Hxc}[P],
  \label{}
\end{equation}
where $H_{L}$ is the linear part of the Hamiltonian, and is a given
matrix derived from the discretized Laplacian operator and the
electron-nuclei interaction potential. $E_{\Hxc}[P]$ consists of the
Hartree, and exchange correlation energy, and is a
nonlinear functional of the density matrix $P$. Moreover, all the
information of the quantum system, including the atomic types and positions,
is given by the electron-nuclei interaction in $H_{L}$.
The ground state energy of KSDFT can be obtained from the following variational problem
\begin{equation}
  E = \underset{\scriptsize\begin{array}{c} P^2 = P, P^* = P \\ \Tr P =
    N_e \end{array} }{\inf} \cE[P], \label{eq:energy_non_linear}
\end{equation}

Analogous to Proposition~\ref{prop:linearks}, the corresponding
Euler-Lagrange equation is
\begin{equation}
  H[P]\Psi =(H_{L} + V_{\Hxc}[P]) \Psi = \Psi \Lambda, \quad P=\Psi\Psi^{*},
  \label{}
\end{equation}
where $(\Psi,\Lambda)$ are the lowest $N_{e}$ eigenpairs of the
nonlinear Hamiltonian $H[P]$, and the functional derivative $V_{\Hxc}[P]
= \frac{\delta E_{\Hxc}[P]}{\delta P}$ is called the
exchange-correlation potential. This is precisely~\eqref{eqn:KSDFT}.
However, we remark that the procedure of taking the lowest $N_{e}$
eigenpairs, which is called the \textit{aufbau} principle in electronic
structure theories, is not always valid.
The \textit{aufbau} principle has been found to be violated
for certain model energy functionals~\cite{LiuWenWangEtAl2015}, but numerical experience indicates
that it generally holds in the context of KSDFT calculations for real
materials.  In the discussion below, we always assume the counterpart to
Proposition~\ref{prop:linearks} holds for the nonlinear problems under
consideration.

According to the discussion in Section~\ref{sec:linear}, the key ansatz
of the PET is that for some reference system with a different linear
part of the Hamiltonian $H_{0,L}$, we have evaluated the density matrix
and computed the projector $P_{0,b}$. Then for the system of interest,
PET evaluates the modified variational problem by restricting the \REV{feasible
set} of the density matrix as
\begin{equation}
  E^{\text{PET}} = \underset{\scriptsize\begin{array}{c} P_s^2 = P_s,
    P_s^* = P_s \\ P_{0,b}P_s = 0,  \Tr P_s = N_s \end{array} }{\inf}
    \cE[P_{s}+P_{0,b}], \label{eq:energy_pet_non_linear}
\end{equation}
Analogous to Proposition~\ref{prop:linearpet}, by assuming the
corresponding \textit{aufbau} principle, PET can be solved by the
following nonlinear eigenvalue problem
\begin{equation}
  H[P]\vert_{\mathcal{B}_{0}^{\perp}}\Psi_{s} := \Psi_{s} \Lambda_{s}, \quad
  P_{s}=\Psi_{s}\Psi_{s}^{*}, \quad P=P_{s}+P_{0,b}.
  \label{eqn:nonlinear_pet}
\end{equation}
Here $(\Psi_{s},\Lambda_{s})$ are the lowest $N_{s}$ eigenpairs of the
self-consistent Hamiltonian $H[P]\vert_{\mathcal{B}_{0}^{\perp}}$.

The first order perturbative correction to PET is entirely analogous to
Proposition~\ref{prop:perturblinearpet}.  According to
Remark~\ref{rem:effectperturb}, the effectiveness of the perturbative
approach mainly lies in the fact that it only needs to be applied once
after ~\eqref{eqn:nonlinear_pet} reaches self-consistency.

\REV{Once $P^{\text{pert}}$ is obtained we define the energy as
\begin{equation}
\label{eqn:E_pert}
  E^{\text{pert}} := E^{\text{PET}} + \Tr[P^{\text{PET}} H[P^{\text{PET}}] P^{\text{pert}}],
\end{equation}
\ie our correction of the energy is only at the linear level.
We point out that \eqref{eqn:E_pert} is only correct in the spinless or
spin unrestricted case. For spin restricted calculations a factor $1/2$
needs to the included in the correction.}

\REV{In addition, we note that we can compute the atomic forces for the PET solution
using the Hellmann-Feynman formula, which is due to the fact that the
solution satisfies a variational principle. However, for the perturbation,
the resulting approximation does not satisfy any variational principle,
thus we use an expensive finite difference approach to compute the forces.
For the sake of consistency we use an standard second order finite
difference scheme to approximate the force for both PET and the corrected
approximation. }

\begin{rem}
  In~\cite{ManbyStellaGoodpasterEtAl2012} the Euler-Lagrange equation
  takes a slightly different form from~\eqref{eqn:nonlinear_pet}. The
  connection with the present formulation can be established
  by noting that the energy functional satisfies the identity
  \[
  \cE[P_{s}+P_{0,b}] = \cE[P_{s}] + \left(\cE[P_{s}+P_{0,b}] -
  \cE[P_{s}] \right) .
  \]
  Then the Euler-Lagrange equation gives the Hamiltonian
  \[
  H_{L}+V_{\Hxc}[P_{s}] + (V_{\Hxc}[P_{s}+P_{0,b}]-V_{\Hxc}[P_{s}])
  \]
  restricted to the subspace $\mathcal{B}_{0}^{\perp}$. The term in the
  parenthesis,
  $V_{\text{emb}}(P_{s}):=(V_{\Hxc}[P_{s}+P_{0,b}]-V_{\Hxc}[P_{s}])$, is
  called the ``embedding potential'', which can be interpreted as an
  external potential imposed onto the system part from the bath. For
  instance, in the absence of the exchange-correlation, $V_{\Hxc}\equiv
  V_{\text{H}}$ is a linear mapping. Then $V_{\text{emb}} =
  V_{\text{H}}[P_{0,b}]$ is the Coulomb interaction solely due to the projector
  from the bath.
\end{rem}

\section{Evaluation of the \REV{bath} projector}\label{sec:implementation}

%In this section we discuss issues concerning the practical implementation
%of PET and the perturbative correction.
%
%\subsection{Localization and \REV{bath} projector}
The success of PET relies on a proper choice of the reference
projector $P_{0,b}$. The suggestion
from~\cite{ManbyStellaGoodpasterEtAl2012} is to compute a set of
localized functions within the subspace $\text{span}\{\Psi_{0}\}$ to
evaluate $P_{0,b}$. For simplicity, we use the
notation from the linear problem, but the procedure can be directly
generalized to the nonlinear setup as well.

Simply speaking, for a class of matrices $H$ satisfying the gap condition,
we may expect that the matrix elements of the density matrix $P$ decays
rapidly along the off-diagonal direction.  In the physics literature
this is referred to as the ``nearsightedness''
principle~\cite{Kohn1996,ProdanKohn2005}, and there is a rich literature
studying the validity of such decay
property (see \eg~\cite{BrouderPanatiCalandraEtAl2007,BenziBoitoRazouk2013}).
We further expect that there exists a unitary matrix $U\in
\CC^{N_{e}\times N_{e}}$, called a gauge matrix, so that each
column of the rotated matrix $\Phi = \Psi U$
is localized, \ie it concentrates on a small number of
elements compared to the size of the vector $N$.
We point out that efficient numerical algorithms have been developed to compute such
gauge and the corresponding localized
functions (see
\eg~\cite{FosterBoys1960,MarzariVanderbilt1997,DamleLinYing2015}).
Once the localized
functions are obtained, we may find localized functions associated with
the index set for the bath $\mathcal{I}_{b}$ denoted by $\Psi_{0,b}$.
To make the discussion self-contained, we briefly introduce the recently
developed selected
columns of the density matrix (SCDM) method~\cite{DamleLinYing2015}
below as a simple and robust localization method to generate
$P_{0,b}$. Other localization techniques can certainly be used as well.

The main idea of the SCDM procedure is that the localized function $\Phi$ are
obtained directly from columns of the density matrix $P = \Psi\Psi^*$.
However, picking $N_e$ random columns of $P$ may result in a poorly
conditioned basis.  In order to choose a well
conditioned set of columns, denoted $\CS = \left\{c_1,c_2,\ldots,c_{N_e}
\right\},$ we may use a QR factorization with column
pivoting (QRCP) procedure~\cite{GolubVan2013}. More
specifically, we compute
\begin{equation}
\label{eqn:qrcp}
\Psi^*\Pi = U\begin{bmatrix} R_1 & R_2
             \end{bmatrix},
\end{equation}
where $\Pi$ is a permutation matrix so that $R_{1}$ is a well
conditioned matrix. The set $\CS$ is given by the union of the nonzero
row indices of the first $N_{e}$ columns of the permutation matrix
$\Pi$.
The unitary matrix $U$ is the desired gauge matrix~\cite{DamleLinYing2015,DamleLinYing2017a}, and
$\Phi=\Psi U$ is a localized matrix. It can be
seen that Eq.~\eqref{eqn:qrcp} directly leads to a QRCP factorization of
$P$ as
\[
P\Pi = \Psi \Psi^{*}\Pi = (\Psi U)\begin{bmatrix} R_1 & R_2
\end{bmatrix},
\]
and $\Psi U$ is a matrix with orthogonal columns.

%denote the columns selected by the first $N_e$ columns of $\Pi$ then
%$A_{:,\CS}$ should be a well conditioned set of $N_e$ columns of $A.$

Let us apply the SCDM procedure to $H_{0}$ and its eigenfunctions
$\Psi_{0}$. With some abuse of notation, from a
pre-defined bath index set $\mathcal{I}_{b}\subset\{1,\ldots,N\}$, we may associate
the $i$-th column of $\Phi_{0}$ to the bath degrees of freedom if the $i$-th
element of $\CS$ is in $\mathcal{I}_{b}$. These selected vectors,
denoted by $\Phi_{0,b}$ form the \REV{bath} projector $P_{0,b}$. Finally, the
condition~\eqref{eqn:Psi0b} can be satisfied by solving the following
eigenvalue problem
\begin{equation}
  \Phi_{0,b}^{*} H_{0}\Phi_{0,b} C_{0,b} = C_{0,b} \Lambda_{0,b},
  \label{eqn:Phieig}
\end{equation}
and then $\Psi_{0,b}=\Phi_{0,b} C_{0,b}$. We summarize the procedure for
computing the $\Psi_{0,b}$ in Algorithm~\ref{alg:psi0b}.

\begin{algorithm}
  \caption{Using the SCDM algorithm for constructing the \REV{bath} projector.}
\label{alg:psi0b}
\begin{algorithmic}[1]
  \Statex Input: $H_{0}$, $\Psi_{0}$, $\mathcal{I}_{b}$,
  $N_{e}^{0}$.
  \Statex Output: $\Psi_{0,b}$.
  \State Perform QRCP for $\Psi_{0}^*$:
  $\Psi_{0}^*\Pi = U\begin{bmatrix} R_1 & R_2 \end{bmatrix}.$
    The set $\CS$ is given by the union
    of the nonzero row indices of the first $N_{e}^{0}$ columns of the
    permutation matrix $\Pi$.
  \State Compute $\Phi_{0} = \Psi_{0} U$. Form a submatrix
  $\Phi_{0,b}:=[\varphi_{i;0}]_{\CS_{i}\in
  \mathcal{I}_{b}}$, where $\varphi_{i;0}$ is the $i$-th column of
  $\Phi_{0}$.
  \State Solve the eigenvalue problem~\eqref{eqn:Phieig}, and compute
  $\Psi_{0,b}=\Phi_{0,b} C_{0,b}$.
\end{algorithmic}
\end{algorithm}

\REV{
\begin{rem}
  We point out that after performing the localization in Alg.~\ref{alg:psi0b},
  the vectors in the resulting bath orbitals, $\Psi_{0,b}$, are not
  eigenvalues of $H_0$. Thus, as shown in the prequel,
  the perturbative correction does not improve the asymptotic convergence rate; however, the
  preconstants are greatly reduced. In fact, as it will be shown in the numerical experiments, when the
  perturbation is relatively large, the perturbative correction associated with the rotated
  vectors $\Psi_{0,b}$ has a considerable smaller error than the one associated to the eigenvectors of $H_0$.
\end{rem}
}
\section{Numerical Examples}\label{sec:numer}
We present several examples to demonstrate the effectiveness of the PET
method and the perturbation scheme. The numerical tests were
coded in Matlab 2017b. For the solution of KSDFT in the nonlinear case,
\REV{PET and the perturbative correction are implemented within} the KSSOLV~\cite{YangMezaLeeEtAl2009}
software package. All calculations are performed in a dual socket
server with Intel Xeon E5-2670 CPU's and 386 Gb of RAM.

\subsection{Linear Case}

%For the linear in 1D we use a simple Hamiltonian given by
We first consider a simple Hamiltonian in 1D
with zero Dirichlet boundary conditions:
\begin{equation}
  H_0 = -\frac{1}{2}\frac{d^2}{dx^2} + V_0(x), \quad  V_0(x) := \sum_{i
  =1}^3 -40e^{-100(x-\tilde{x}_i)^2}, \quad x \in [-1, 1].
\end{equation}
Here the centers of the Gaussians $\tilde{x} = (-0.5, 0, 0.5)^{T}$. The
1D Laplacian is discretized with a standard 3- point stencil finite
difference scheme with $512$ grid points.

For the reference problem, we evaluate the $3$ eigenfunctions
corresponding to the lowest 3 eigenvalues.
As shown in Fig.~\ref{fig:orbitals}, the
eigenvectors $\Psi_0$  are indeed delocalized across the entire interval
$[-1,1]$. After applying the SCDM algorithm (see Alg.~\ref{alg:psi0b}),
the resulting orbitals $\Phi_{0}$ become much more localized as shown in
Fig.~\ref{fig:orbitals}.

\begin{figure}
\centering
 % \vspace{-.3cm}
\includegraphics[trim= 11mm 0mm 0mm 0mm,clip, width=8cm]{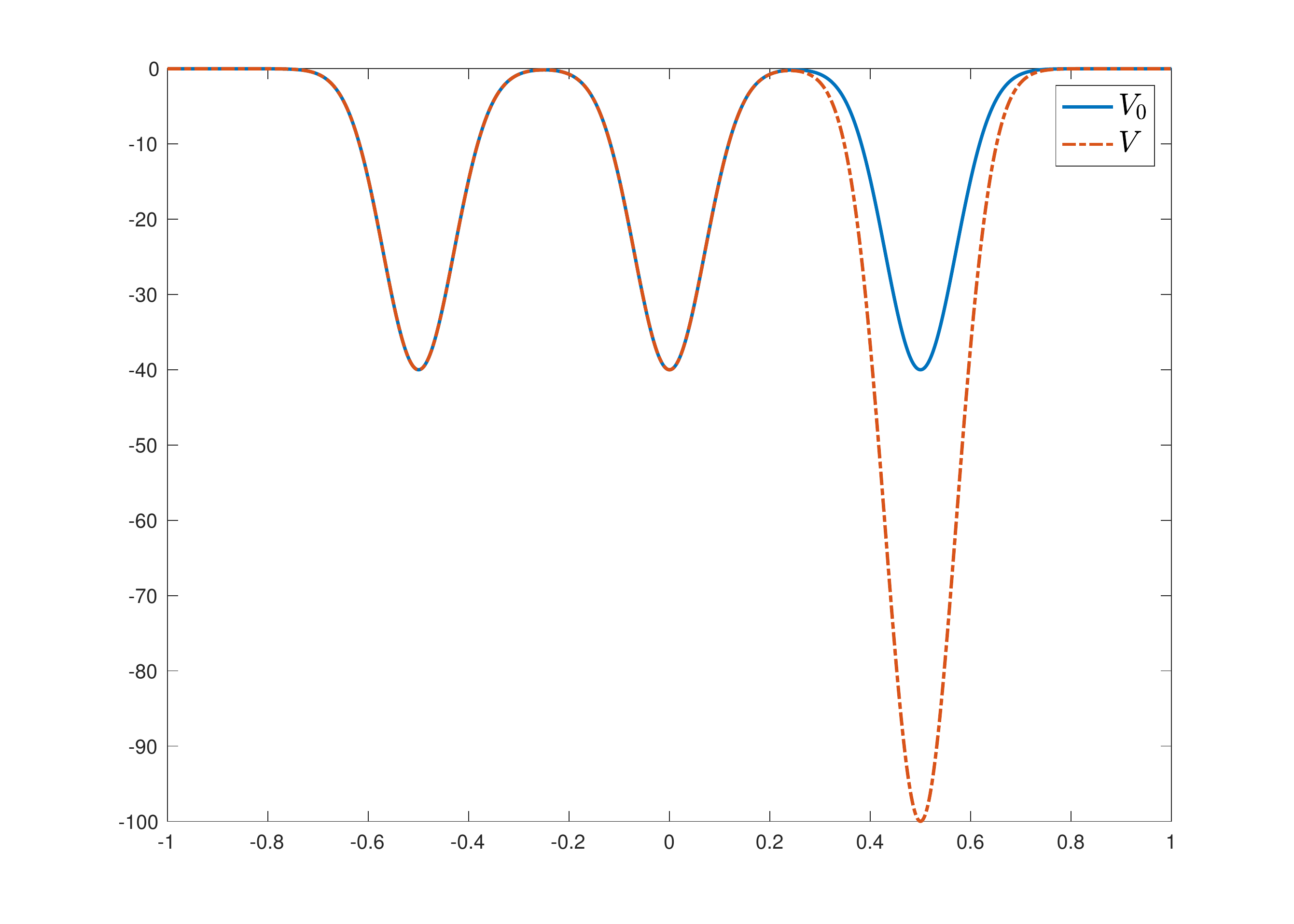}
% \vspace{-.3cm}
\caption{Potential for both $H_0$ and $H$.} \label{graph:1D_potentials}
\vspace{-.3cm}
\end{figure}

\begin{figure}
\centering
 % \vspace{-.3cm}
\subfloat[]{\includegraphics[trim= 20mm 10mm 20mm 10mm,clip, width=6.5cm]{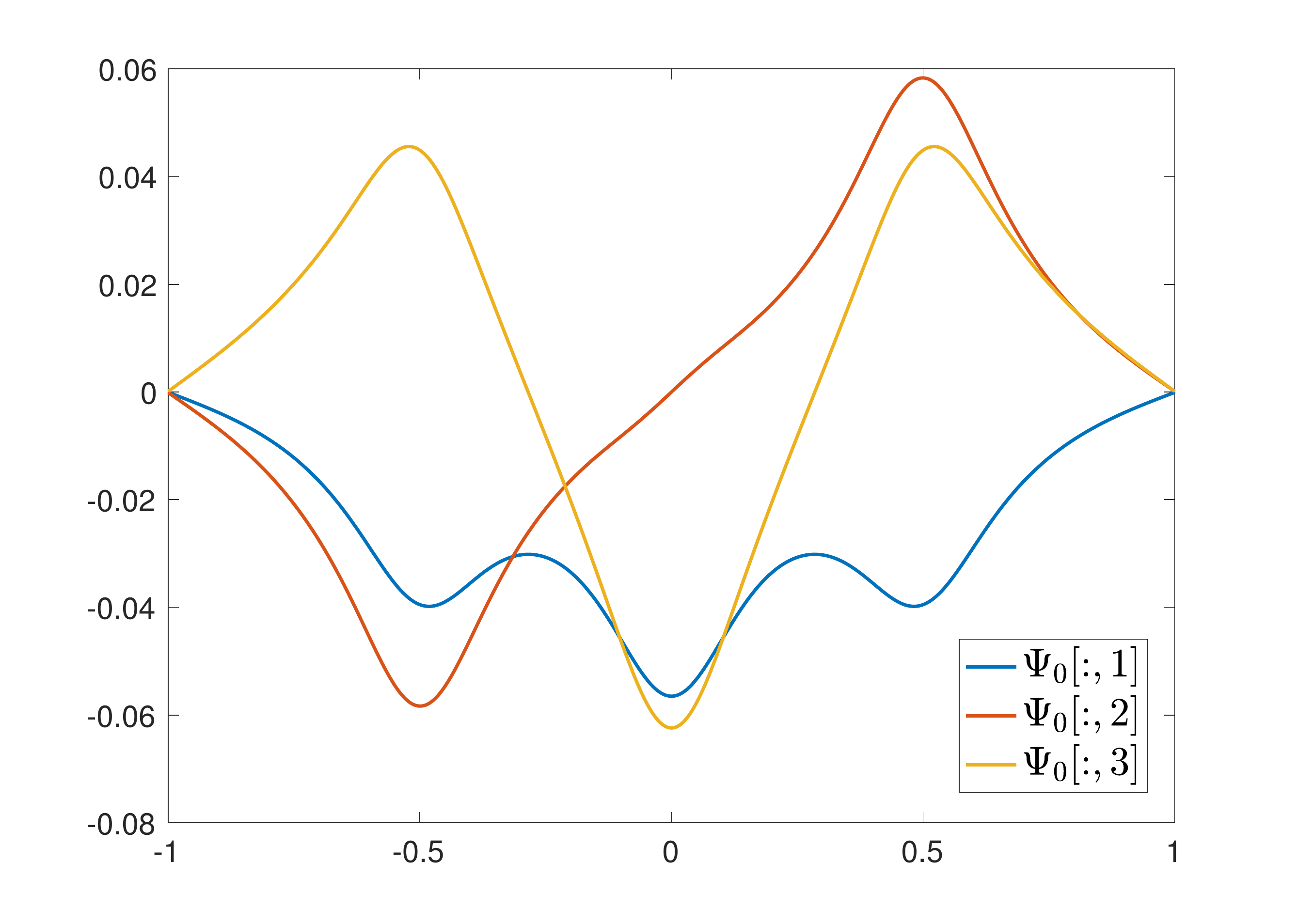}}
\subfloat[]{\includegraphics[trim= 20mm 10mm 20mm 10mm,clip, width=6.5cm]{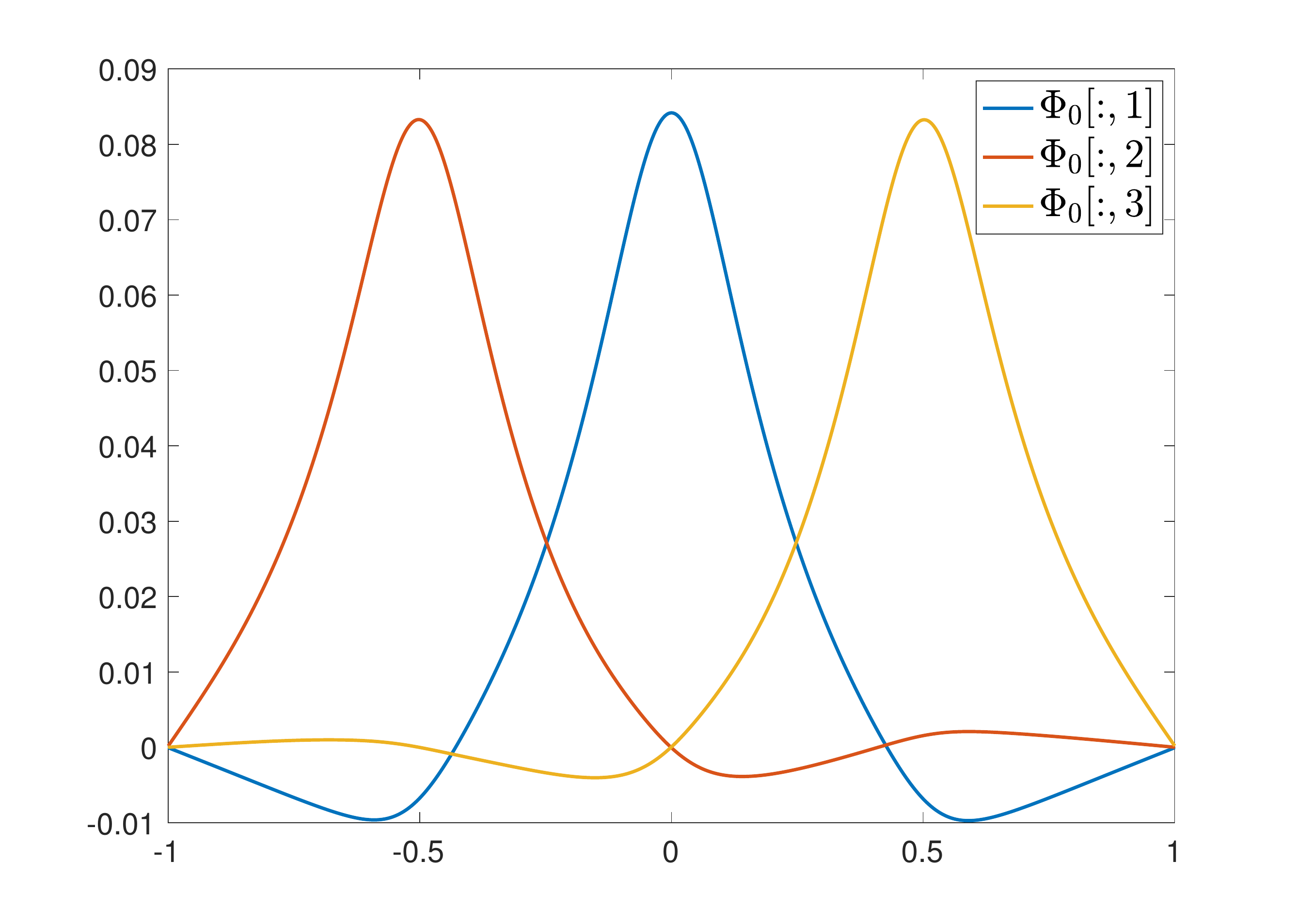}}
% \vspace{-.3cm}
\caption{(a) First $3$ delocalized orbitals (columns of $\Psi_0$); (b) first $3$ localized orbitals (columns of $\Phi_0$). } \label{fig:orbitals}
\vspace{-.3cm}
\end{figure}

We define the new Hamiltonian by changing the height of the last
Gaussian function as
\begin{equation}
  H = -\frac{1}{2}\frac{d^2}{dx^2} +  V(x), \quad  V(x) : = \sum_{i =1}^2 -40e^{-100(x-\tilde{x}_i)^2} -100e^{-100(x-\tilde{x}_3)^2}.
\end{equation}

We observe that two columns in $\Phi_{0}$ are
localized far from the modified Gaussian,
and we consider them as the bath orbitals. We set $\mathcal{I}_{b} =
\{1,\ldots,  340\}$, this ensures $\Phi_{0,b} =
[ \Phi_0[:, 1], \Phi_0[:, 2] ] $.
We then compute the linear PET problem to obtain $\Psi_s$, and we build the PET
density matrix as shown in Fig.~\ref{fig:approx_density_matrices_1d}, which
is accurate up to 3 digits in relative error. Furthermore, the error is mostly
localized around the third Gaussian function as one would expect.
The relative error of the energy is $1.42\times 10^{-3}$.  We find it
remarkable that for such a small system, the solution from PET is
already very accurate despite the strong overlap of the system and bath
orbitals.

Finally, we use Alg.~\ref{alg:perturbation} to compute the perturbed density matrix,$P^{\text{pert}}$, which is
more accurate than the PET density matrix without the perturbation, $P^{\text{PET}}$, as
depicted in Fig.~\ref{fig:density_matrices_1d}. We can observe how the
perturbation decreases the error in the density matrix by taking a look
at the \REV{electron} density, $\rho = \diag(P)$, in
Fig.~\ref{fig:err_rho}. In addition, the accuracy of the energy is
improved, with its relative error reduced from  $1.42\times 10^{-3}$ to $1.01\times 10^{-4}$.
If we increase the bath size from $1$ to $2$, the accuracy of the energy
is improved further to $7.15\times 10^{-5}$ and $2.12\times 10^{-5}$, without and with
the perturbative correction, respectively.
%Fig.~\ref{fig:err_rho}. In addition, the accuracy of the energy is
%improved from an relative error of $1.42\times 10^{-3}$ to $1.01e-04$.
%In addition, it is possible to increase the bath size to $2$, in such a case the approximation of the energy is improved from $7.15e-05$ to $2.12e-05$.

\begin{figure}
\centering
 % \vspace{-.3cm}
\subfloat[]{\includegraphics[trim= 20mm 70mm 15mm 60mm,clip, width=6.25cm]{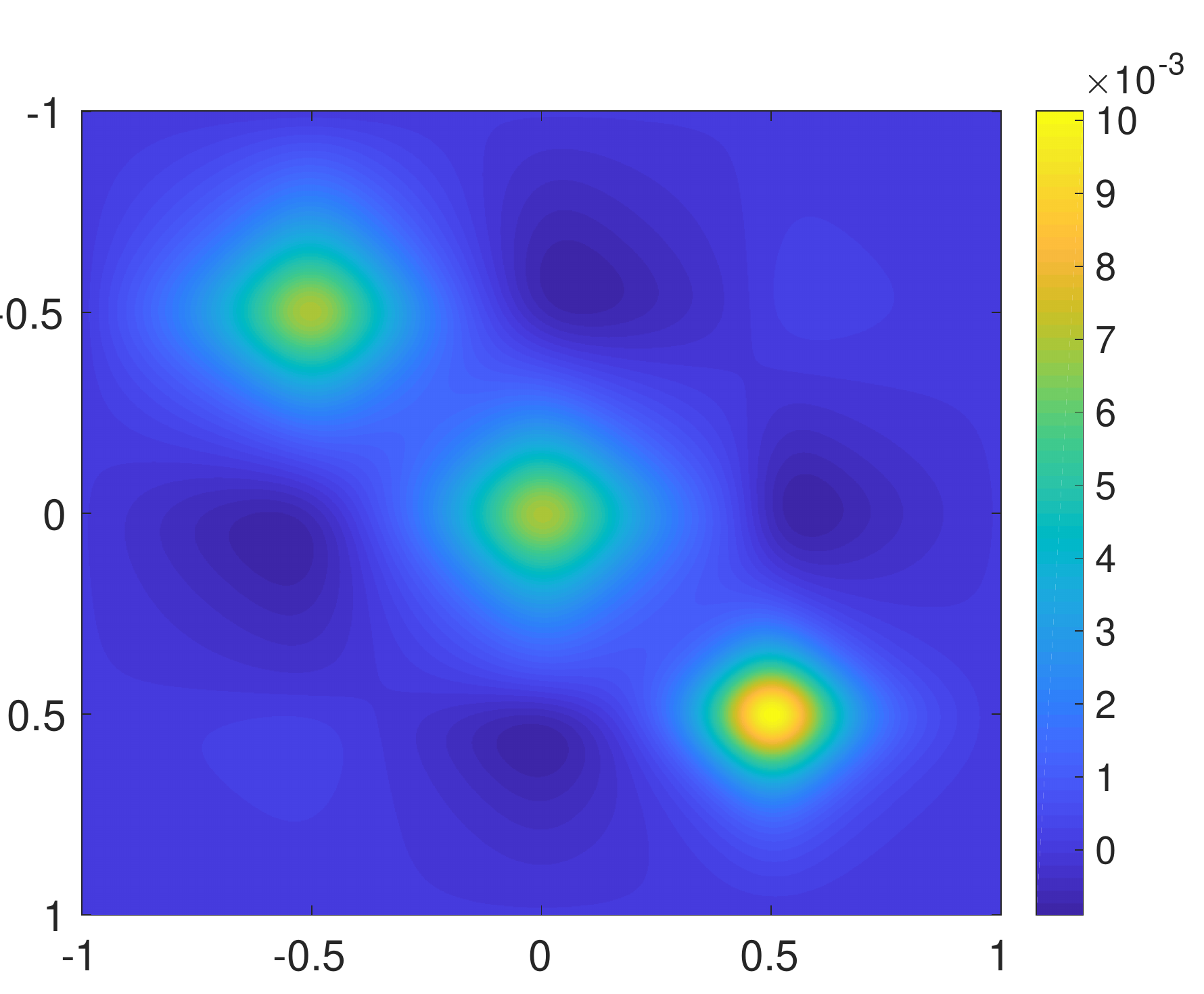}}
\subfloat[]{\includegraphics[trim= 20mm 70mm 15mm 60mm,clip, width=6.25cm]{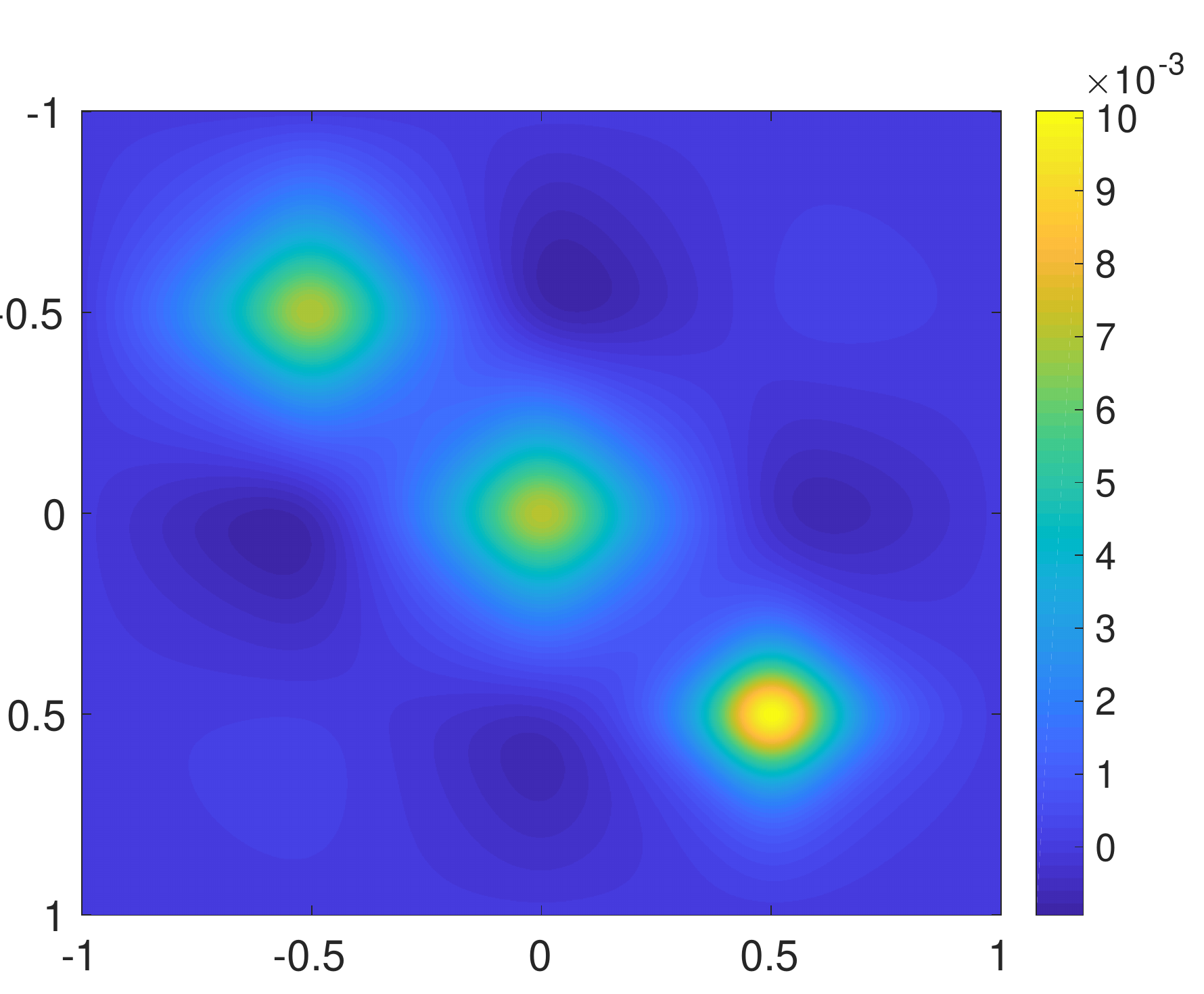}}
% \vspace{-.3cm}
\caption{(a) Exact density matrix; (b) PET density matrix} \label{fig:approx_density_matrices_1d}
\vspace{-.3cm}
\end{figure}

\begin{figure}
\centering
 % \vspace{-.3cm}
\subfloat[]{\includegraphics[trim= 20mm 70mm 15mm 60mm,clip, width=6.25cm]{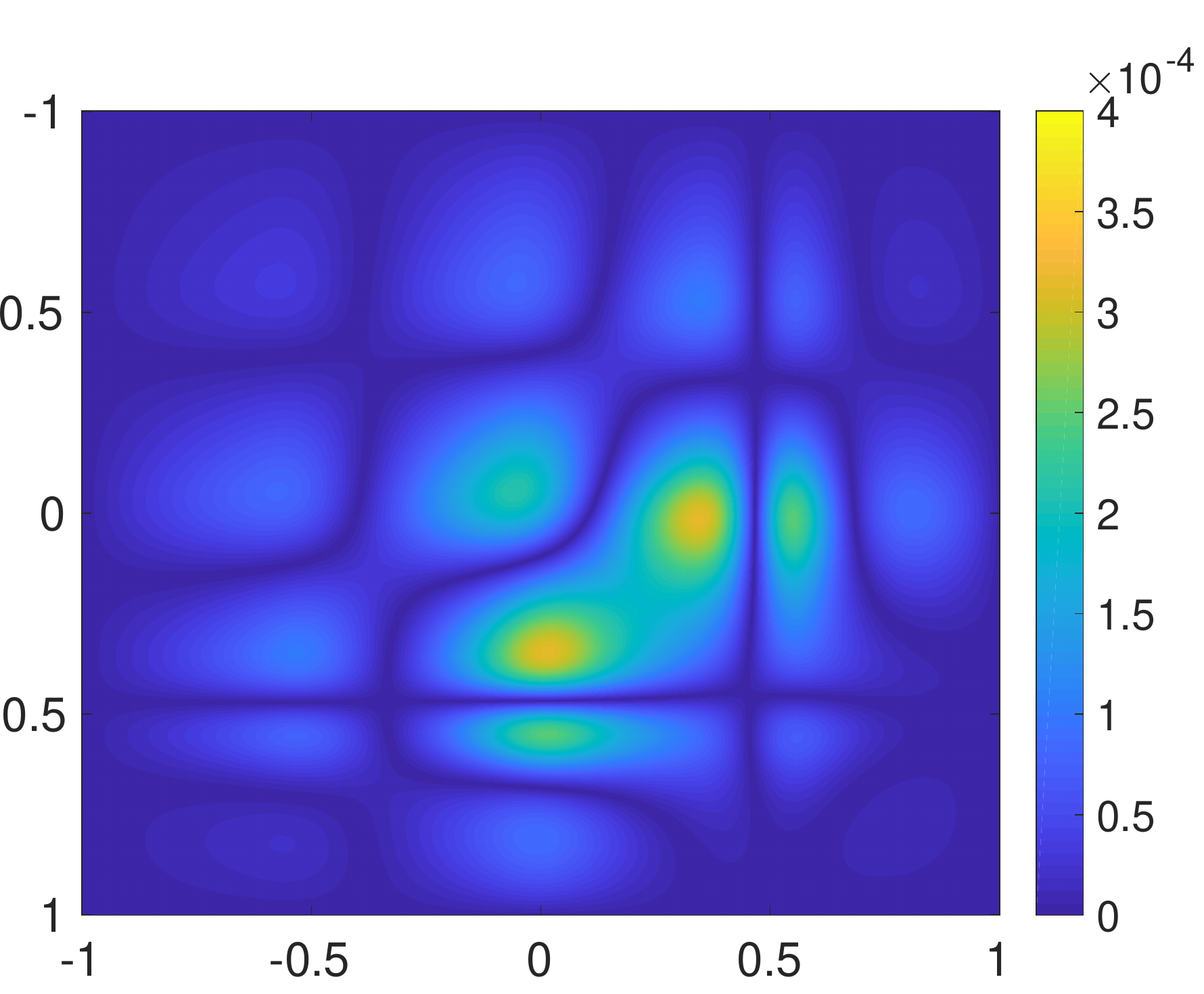}}
\subfloat[]{\includegraphics[trim= 20mm 70mm 15mm 60mm,clip, width=6.25cm]{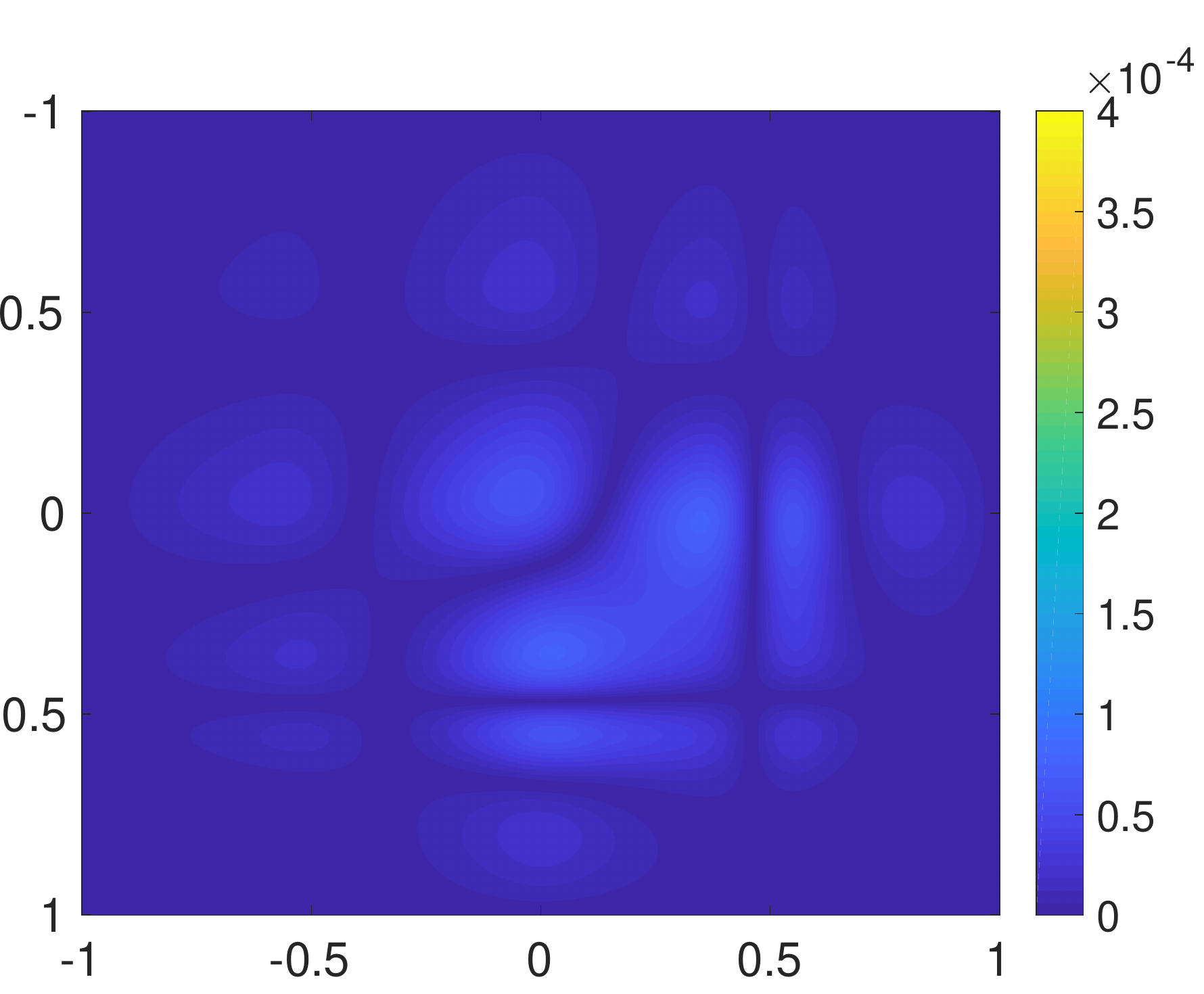}}
% \vspace{-.3cm}
\caption{Error of the density matrix with respect to the reference
answer: (a) PET density matrix; (b) PET density
matrix plus the first order perturbative correction.}
\label{fig:density_matrices_1d}
\vspace{-.3cm}
\end{figure}

\begin{figure}
\centering
\includegraphics[trim= 0mm 0mm 0mm 0mm,clip, width=9.5cm]{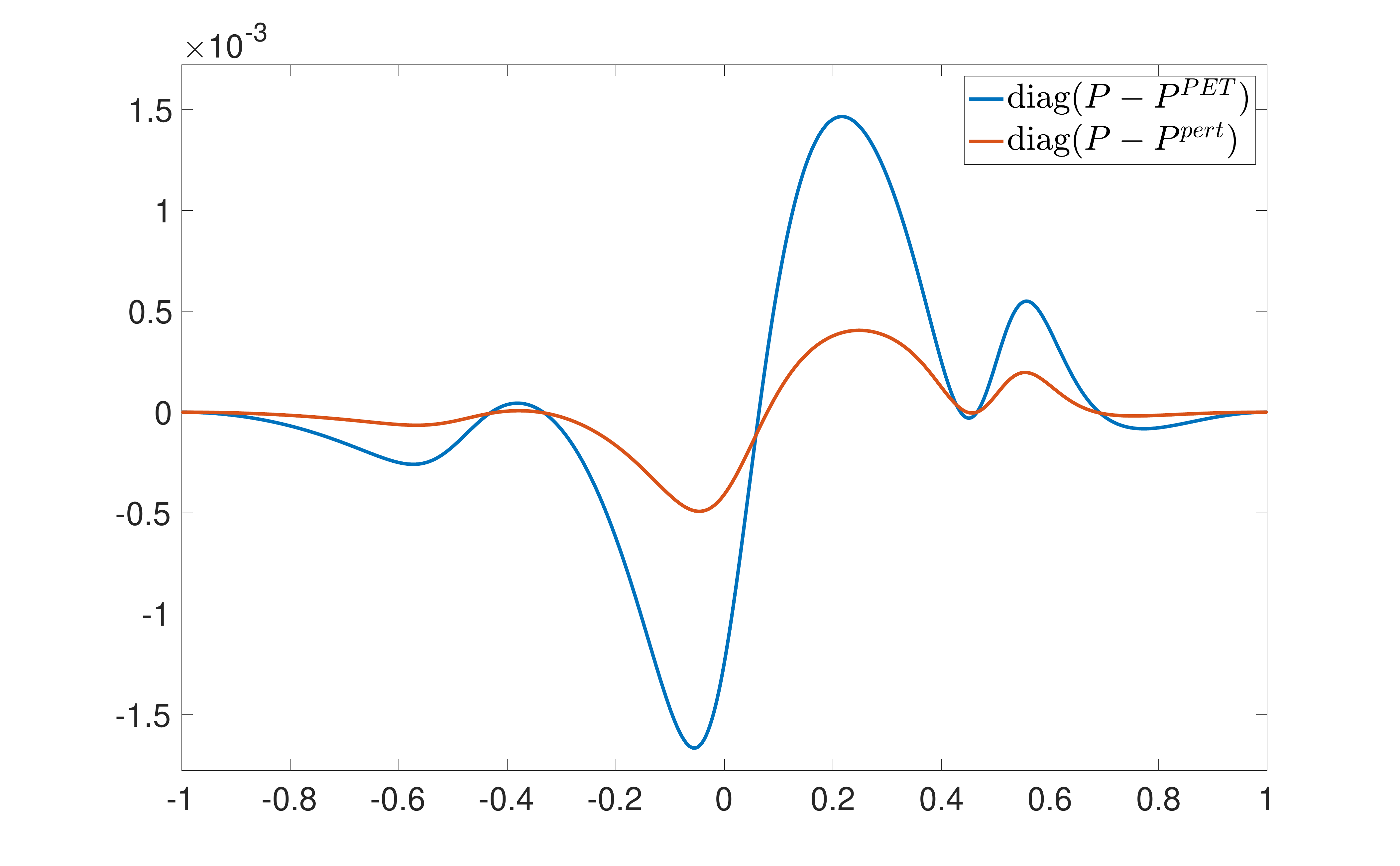}
\caption{Relative error of the electron density using PET density matrix
and the perturbed PET density matrix.} \label{fig:err_rho}
\end{figure}

\REV{In order to showcase the asymptotic convergence of PET and the first
order perturbation discussed at the end of
Section~\ref{sec:perturbationlinear},
we introduce a family of perturbed Hamiltonians as
\begin{equation}
  H_{\epsilon} = -\frac{1}{2}\frac{d^2}{dx^2} +  V_{\epsilon}(x), \quad  V_{\epsilon}(x) : = \sum_{i =1}^2 -40e^{-100(x-\tilde{x}_i)^2} -(40 + \epsilon)e^{-100(x-\tilde{x}_3)^2}.
\end{equation}
Then, in an analogous fashion as above, we compute the PET
approximation and the associated perturbative correction for each
Hamiltonian as $\epsilon \rightarrow 0$. Fig.~\ref{fig:err_PET_pert}
(a) shows the error of approximation of the density matrix and the
energy as the $\epsilon$ tends to zero. As discussed in
Section~\ref{sec:perturbationlinear}, our first order perturbation
is computed with respect to $\delta H^{\text{PET}}_W$, which can
remain to be of $\Or(1)$ even if $\delta H = 0$.}

\REV{On the one hand, when we use localized orbitals to define the system and bath orbitals, Eq.~\eqref{eqn:special_case_condition} is not satisfied. In this case the error of the PET density
matrix and energy decay as $\Or(\epsilon)$ and $\Or(\epsilon^2)$,
respectively as shown by Fig.~\ref{fig:err_PET_pert} (a). Although, the
asymptotic convergence after first order correction remains unchanged,
the preconstants are significantly reduced by one to two orders of
magnitude compared to the results of the PET.}

\REV{On the other hand, when we use the delocalized eigenfunctions to
define the system and bath orbitals, Eq.~\eqref{eqn:special_case_condition} is satisfied.
In such a case, Fig.~\ref{fig:err_PET_pert} (b) shows that the error of the approximate density matrix and energy after the
perturbation correction decay as $\Or(\epsilon^2)$ and $\Or(\epsilon^3)$,
which agrees with results from the standard RS perturbation theory. However,
the preconstants are larger than those in
Fig.~\ref{fig:err_PET_pert} (a). In particular, we can observe from
Fig.~\ref{fig:err_PET_pert} that when the perturbation is relatively large,
partitioning the system with spatially localized orbitals indeed improves
the accuracy of PET, specially when the perturbative correction is used.
}

\begin{figure}
\centering
\subfloat[]{\includegraphics[trim= 10mm 0mm 15mm 15mm,clip,
width=6.3cm]{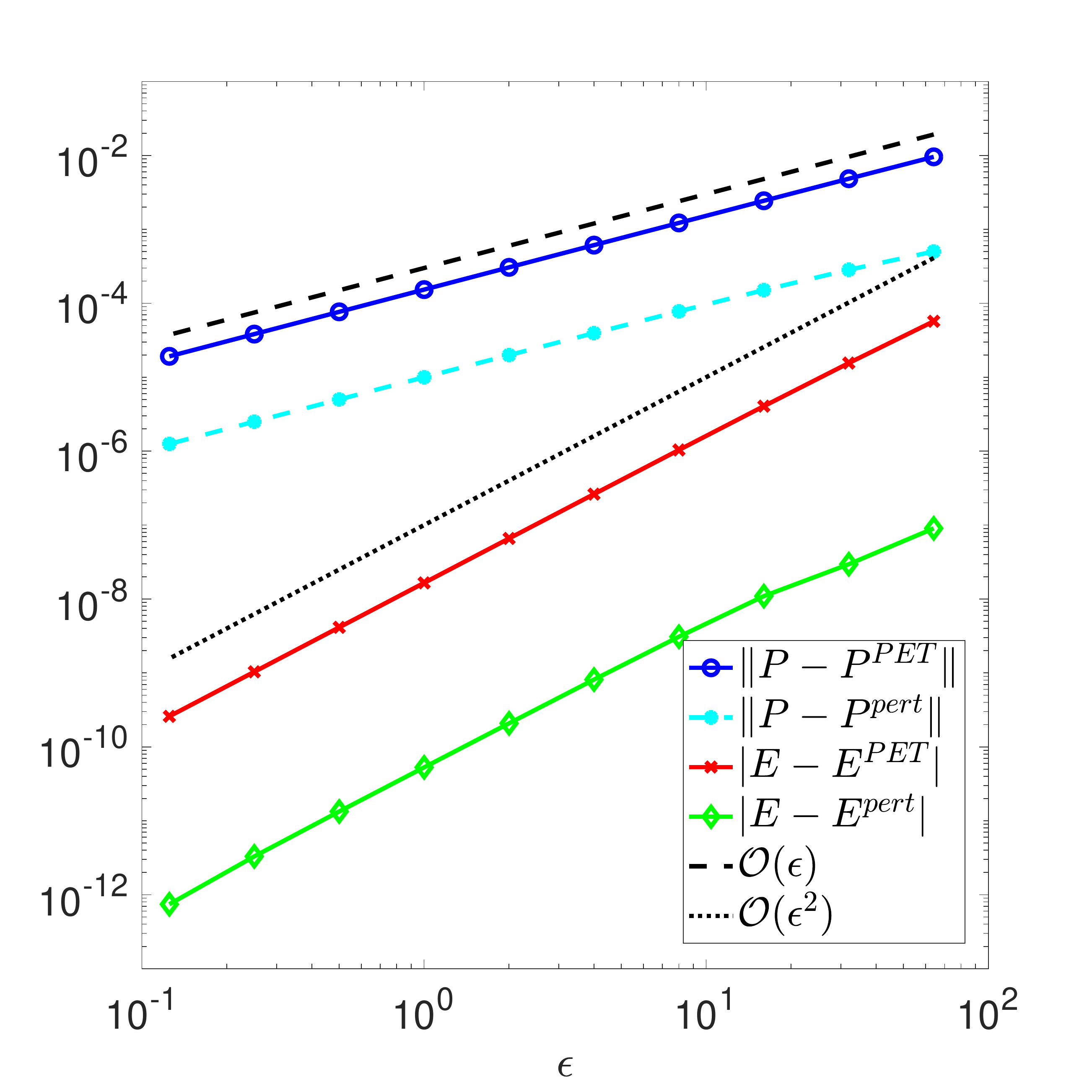}}
\subfloat[]{\includegraphics[trim= 10mm 0mm 15mm 15mm,clip,
width=6.3cm]{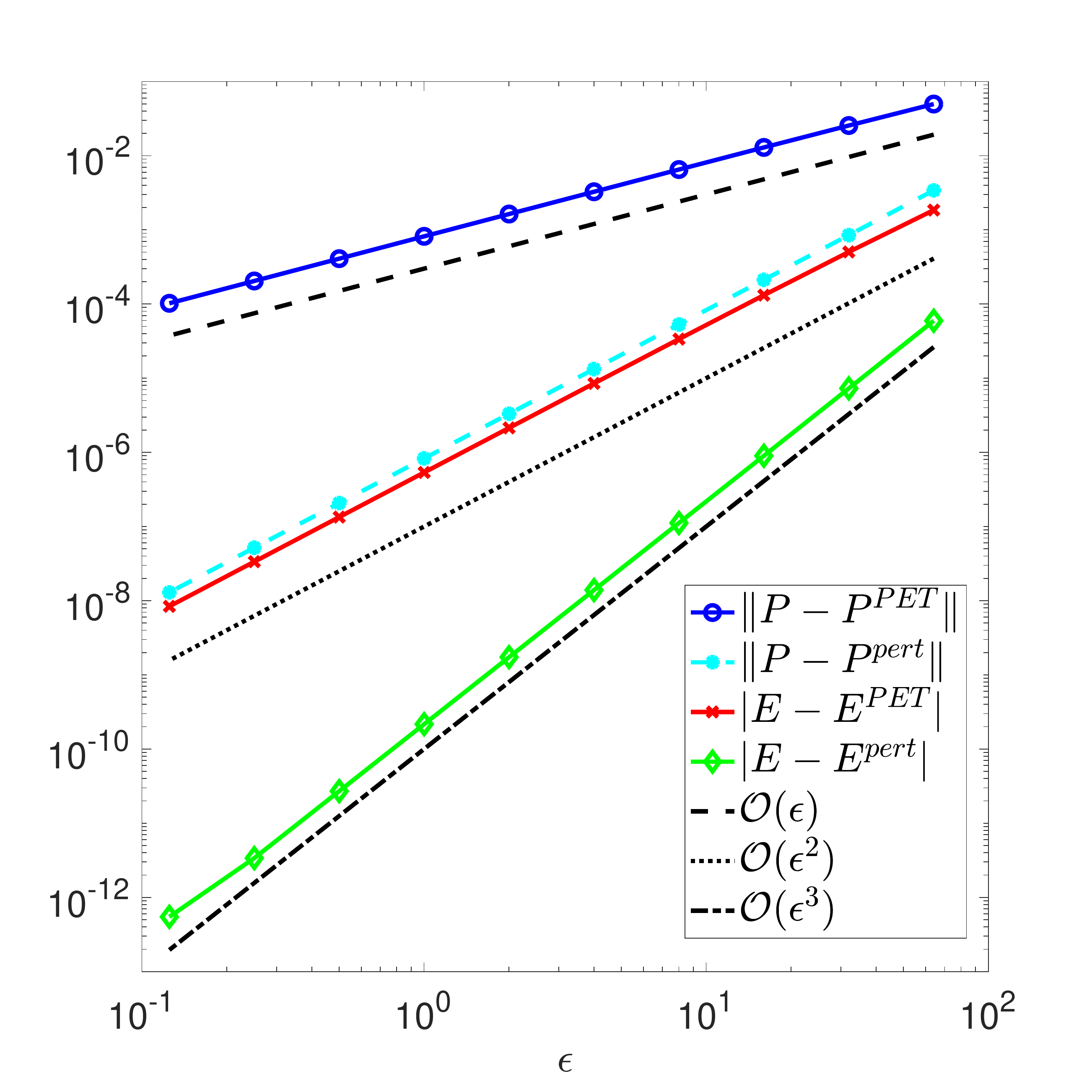}}
\caption{\REV{Error of the \REV{approximate} density matrix and
energy, with bath orbitals defined using (a)
localized orbitals, and (b) delocalized eigenfunctions.}} \label{fig:err_PET_pert}
\end{figure}

\subsection{Nonlinear Case}
For KSDFT calculations, we modified the KSSOLV software
package~\cite{YangMezaLeeEtAl2009} to solve the PET
equations~\eqref{eqn:nonlinear_pet} and to obtain the perturbation
correction. KSSOLV uses a pseudo spectral discretization with the plane
wave set.  All the operators, including Hamiltonian and projection
operators, are efficiently implemented in a matrix-free fashion. Within
each self-consistent field iteration, we use the locally optimal block
preconditioned conjugate gradient method (LOBPCG)~\cite{Knyazev2001} to
solve the linear eigenvalue problems. For the perturbative correction,
we use the GMRES~\cite{SaadSchultz1986} method
with a preconditioner \cite{TeterPayneAllan1989} implemented via fast Fourier
transforms (FFTs).

\subsubsection{Silane}

We first consider a simple molecule, silane (SiH$_{4}$), whose electron
density in shown in Fig.~\ref{fig:SiH4} and we performed three different
numerical experiments to showcase the accuracy of the method. Our
reference system is the silane molecule from an equilibrium configuration.
The bath-system partition is shown in Fig.~\ref{fig:SiH4}, in which
we can observe that we fixed three orbitals as the bath, \REV{induced} by $\mathcal{I}_b$,
and the system part, which is delimited by a pointed red line is considered
as the forth orbital \REV{induced} by $\mathcal{I}_s$. We performed three different
modifications to the atom associated with the fourth orbital:

\begin{figure}
\centering
\subfloat[]{\includegraphics[trim= 0mm 0mm 0mm 0mm,clip, width=6.0cm]{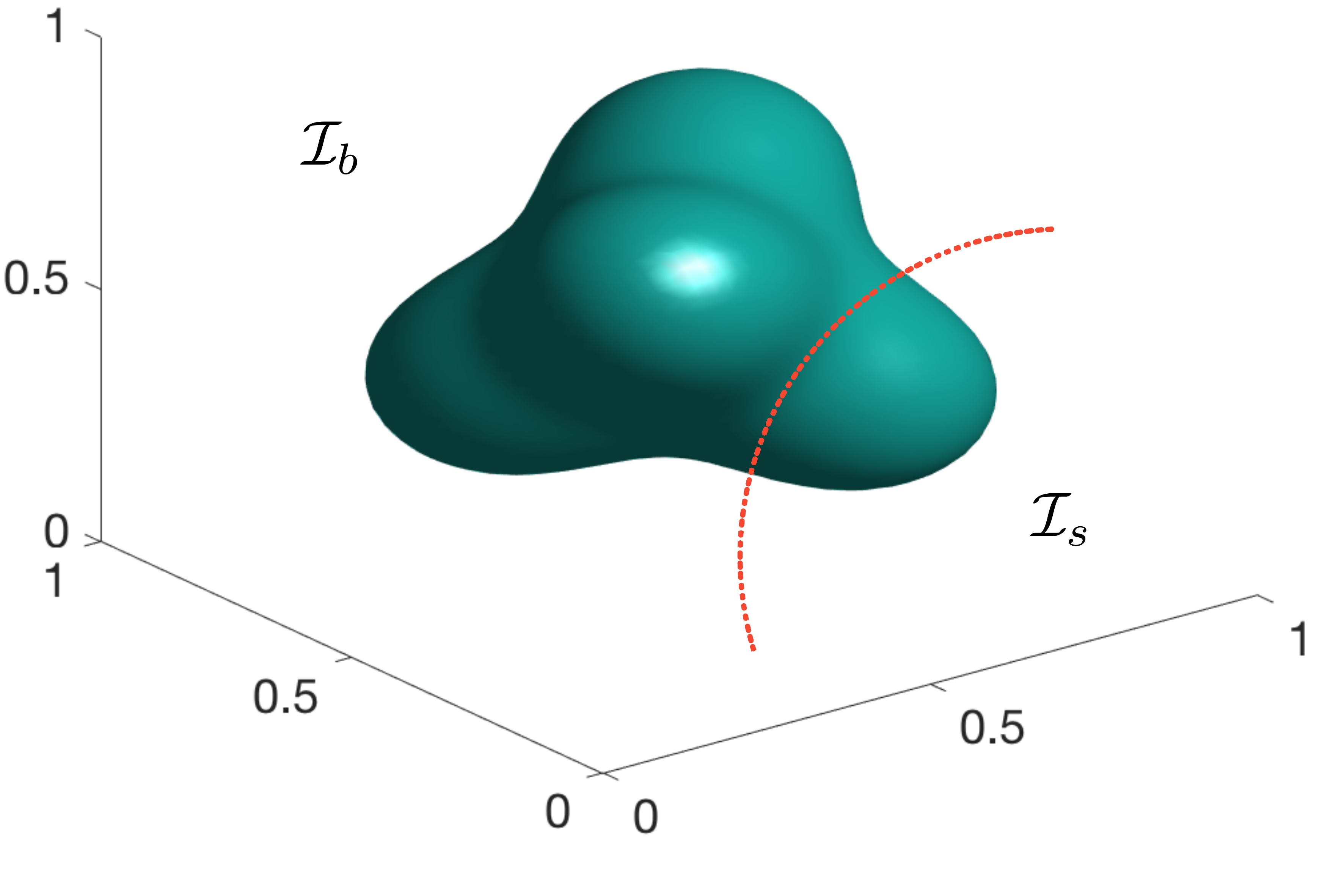}}
\subfloat[]{\includegraphics[trim= 0mm 0mm 0mm 0mm,clip, width=6.0cm]{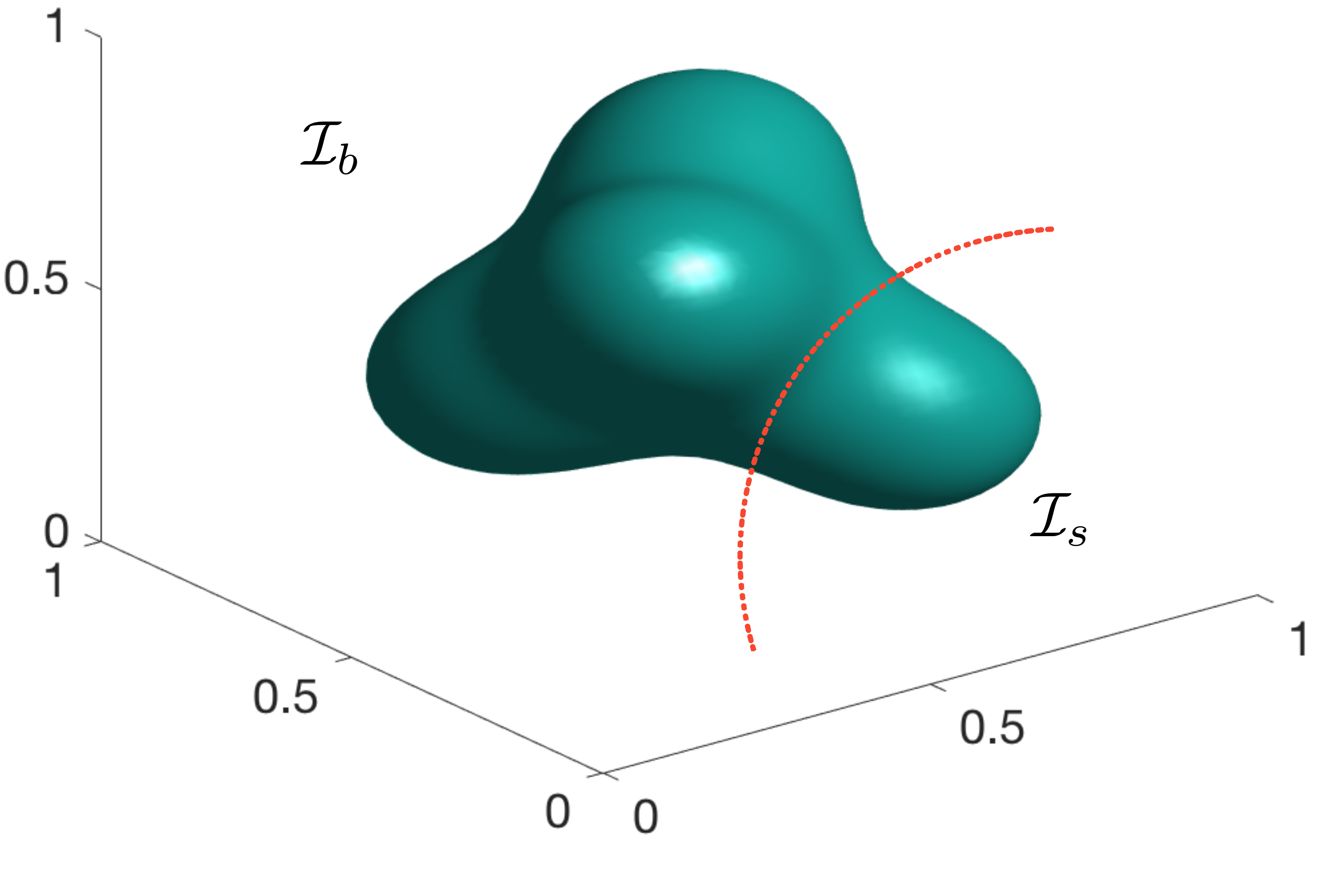}}\\
\subfloat[]{\includegraphics[trim= 0mm 0mm 0mm 0mm,clip, width=6.0cm]{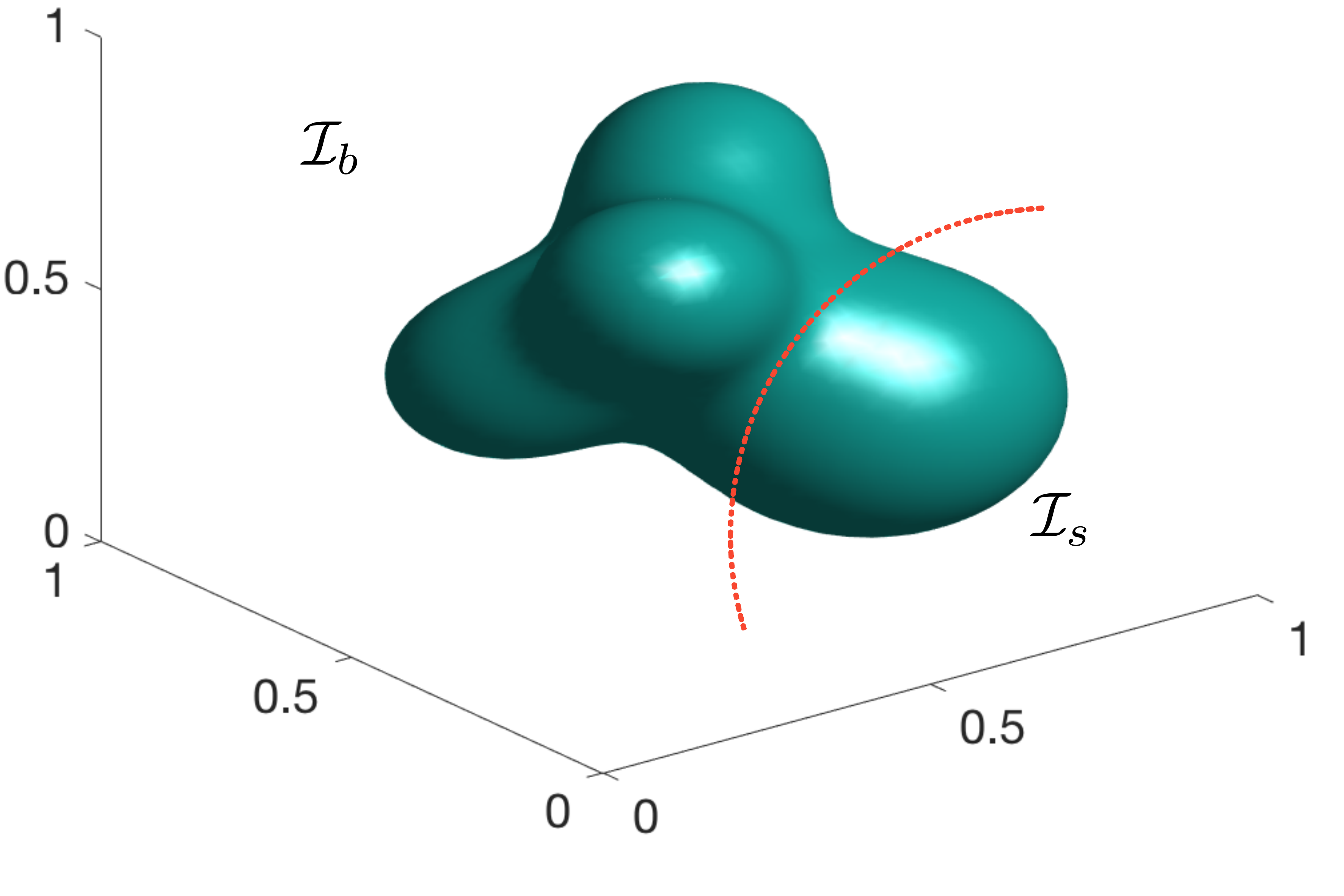}}
\subfloat[]{\includegraphics[trim= 0mm 0mm 0mm 0mm,clip, width=6.0cm]{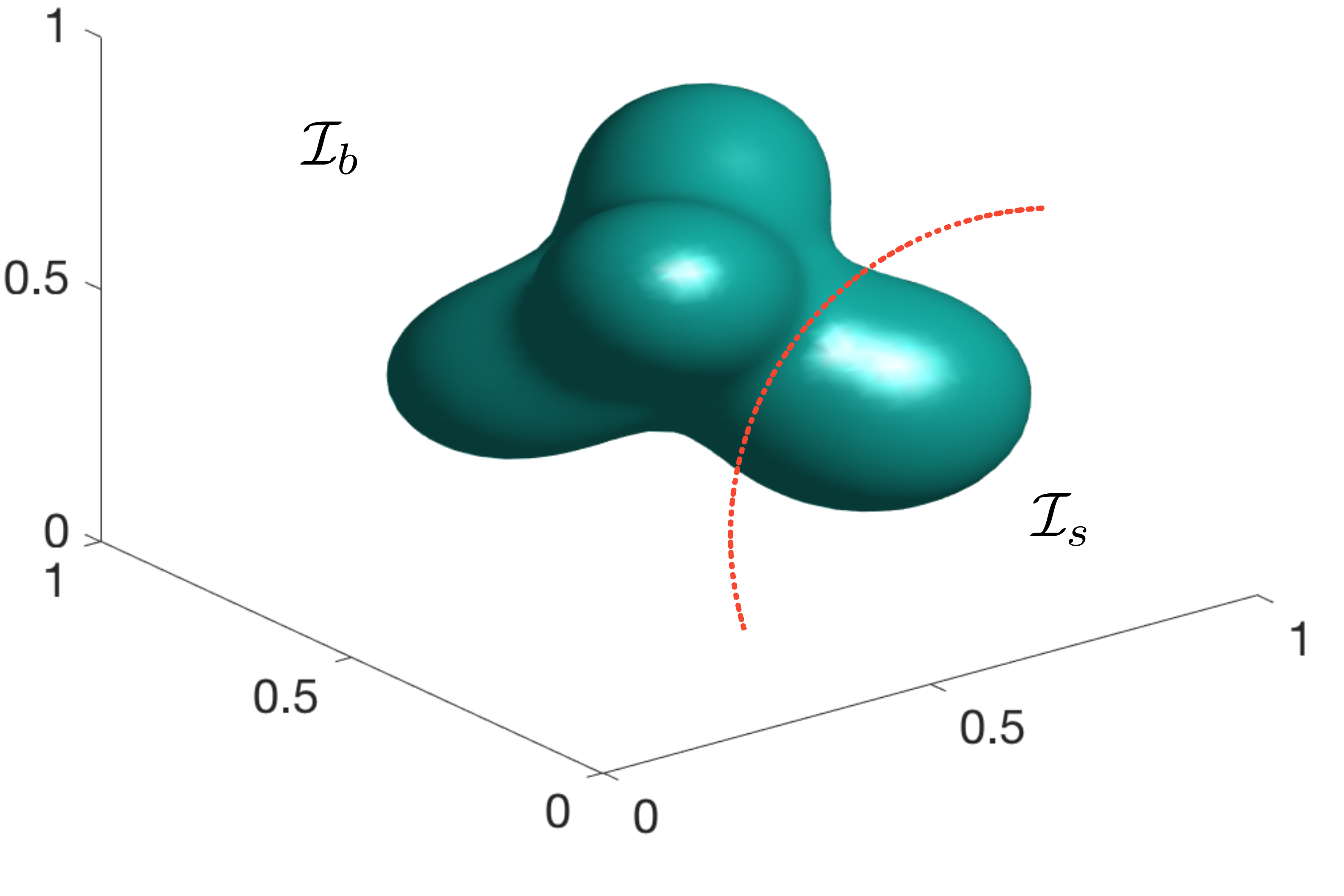}}
\caption{(a) Electron density of the silane molecule, (b) electron density of the SiH$_4$ molecule with one hydrogen bond elongated, (c) electron density of SiH$_3$Cl, and, (d) electron density of SiH$_3$F. } \label{fig:SiH4}
\end{figure}

\begin{itemize}
\item we elongate one hydrogen bond by $25\%$,
\item we replace a hydrogen atom by a chlorine atom (Cl),
\item we replace a hydrogen atom by a fluorine atom (F).
\end{itemize}
\REV{Note that in the last two examples, the number of valence orbitals
in the reference system is $4$, while the number of valence orbitals
in the perturbed systems are both $7$.} Hence the perturbation
introduced by the atom substitution is very large, especially
compared to the small size of the molecule under study here.

We compare the results from PET and the perturbed PET against a
reference solution obtained directly by solving the system in KSSOLV. In particular, we examine the relative error of the density matrices,
the relative error of the \REV{electron} density,  the absolute error of the energy, and
the absolute error of the atomic force at the modified location. All
results are reported in atomic units. In particular, the unit of the energy is
hartree, and the unit of the atomic force is hartree / bohr.
In this case, the energy for PET was computed using the functional in \eqref{eq:energy_pet_non_linear}.
For the perturbed solution, \REV{we used \eqref{eqn:E_pert}. We used a second order finite difference scheme
to compute the forces at the perturbed atom.}

The results for each of the experiments are shown in Tables
\ref{table:numerical_results_P_rho_SiH4} and
\ref{table:numerical_results_E_F_SiH4}. We can observe that the
perturbation effectively reduces the error of the density matrix,
the electron density, the energy and atomic force.
The only exception is the force of SiH$_3$F, which becomes
coincidentally accurate for the PET, but the error after applying the
perturbation theory is still around $10^{-3}$ au. Even for such a small
system, after applying the perturbation formula, the error of the energy
and force already reaches chemical accuracy.
\begin{table}
    \begin{center}
        \begin{tabular}{|c|c|c|c|c|}
            \hline
            Experiment & $ \frac{|| P - P^{\text{PET}}||}{||P||}$
                       & $\frac{|| P - P^{\scriptsize{\mbox{pert}}} ||}{||P||}$
                       & $\frac{||\rho^{\text{PET}} - \rho||}{||\rho||} $
                       &  $\frac{||\rho^{\scriptsize{\mbox{pert}}} - \rho||}{||\rho||}$ \\
            \hline
            Elongated   & $6.03\times 10^{-2}$  & $1.40\times 10^{-2}$ & $1.24\times 10^{-2}$  &  $8.00\times 10^{-3}$   \\
            SiH$_3$Cl   & $7.70\times 10^{-2}$  & $1.74\times 10^{-2}$ & $1.51\times 10^{-2}$  &  $6.71\times 10^{-3}$   \\
            SiH$_3$F    & $9.12\times 10^{-2}$  & $2.09\times 10^{-2}$ & $6.64\times 10^{-3}$  &  $4.72\times 10^{-3}$  \\
            \hline
        \end{tabular}
    \end{center}
    \caption{ Errors of the density matrices and electron densities for
    the different perturbation of the SiH$_4$ molecule.}   \label{table:numerical_results_P_rho_SiH4}
    \vspace{-.3cm}
\end{table}

\begin{table}
    \begin{center}
        \begin{tabular}{|c|c|c|c|c|}
            \hline
            Experiment & $ E - E^{\text{PET}}$
                       & $ E - E^{\scriptsize{\mbox{pert}}}$
                       & $ F - F^{\text{PET}}$
                       & $ F - F^{\scriptsize{\mbox{pert}}}$  \\
            \hline
            Elongated   & $5.98\times 10^{-3}$  & $2.29\times 10^{-4}$ & $1.51\times 10^{-2}$ & $1.33\times 10^{-3}$   \\
            SiH$_3$Cl   & $1.84\times 10^{-2}$  & $1.94\times 10^{-3}$ & $1.89\times 10^{-2}$ & $2.72\times 10^{-3}$   \\
            SiH$_3$F    & $1.66\times 10^{-2}$  & $1.33\times 10^{-3}$ & $9.29\times 10^{-5}$ & $1.13\times 10^{-3}$   \\
            \hline
        \end{tabular}
    \end{center}
    \caption{ Errors for the different perturbation of the SiH$_4$ molecule.}   \label{table:numerical_results_E_F_SiH4}
    \vspace{-.3cm}
\end{table}

\subsubsection{Benzene}

In this example we show the performance of the method for a benzene
molecule (C$_6$H$_6$), whose \REV{electron} density is shown in Fig.~\ref{fig:rho_bencene} (a).
We substitute one of hydrogen atoms by a
fluorine atom, \REV{whose electron}  density is shown in Fig.~\ref{fig:rho_bencene} (b).
The benzene molecule has a total of 15 valence orbitals.
\REV{ To determine the partitions, we created a sphere centered at the replaced atom
and we performed the localization using Alg.~\ref{alg:psi0b} where we labeled the different
localized orbitals depending on the position of their associated pivots (from Alg.~\ref{alg:psi0b}).
In particular, we labeled the orbitals whose pivots were within the sphere as system orbitals, and
the rest as bath orbitals.
The Tables \ref{table:numerical_results_P_rho_bencene} and
\ref{table:numerical_results_E_F_bencene} were generated by incrementally increasing
the radius of the sphere, until obtaining $N_s$ system orbitals}
. The different partitions are depicted in Fig.~\ref{fig:rho_bencene}, for $N_s =1,4$
and $6$, in which the segmented red line indicates the boundary between
the bath and system partitions.
Tables \ref{table:numerical_results_P_rho_bencene} and
\ref{table:numerical_results_E_F_bencene} show the errors of the density
matrix and the electron density, as well as the energy and the atomic
force. \REV{We can observe a systematically decrease on the errors as the bath
size decreases, \ie $N_s$, the system size increases.}
When the system size is $7$, the error of the energy and force after
perturbative correction is already below the chemical accuracy and is as small as $1.02\times 10^{-4}$ and $4.17\times 10^{-4}$ au, respectively.
%, which is below the requirement of the chemical accuracy
%(i.e.  around $10^{-3}$ au for the error of the energy and force, respectively).

\begin{figure}
\centering
\subfloat[]{\includegraphics[trim= 0mm 0mm 0mm 0mm,clip, width=6.0cm]{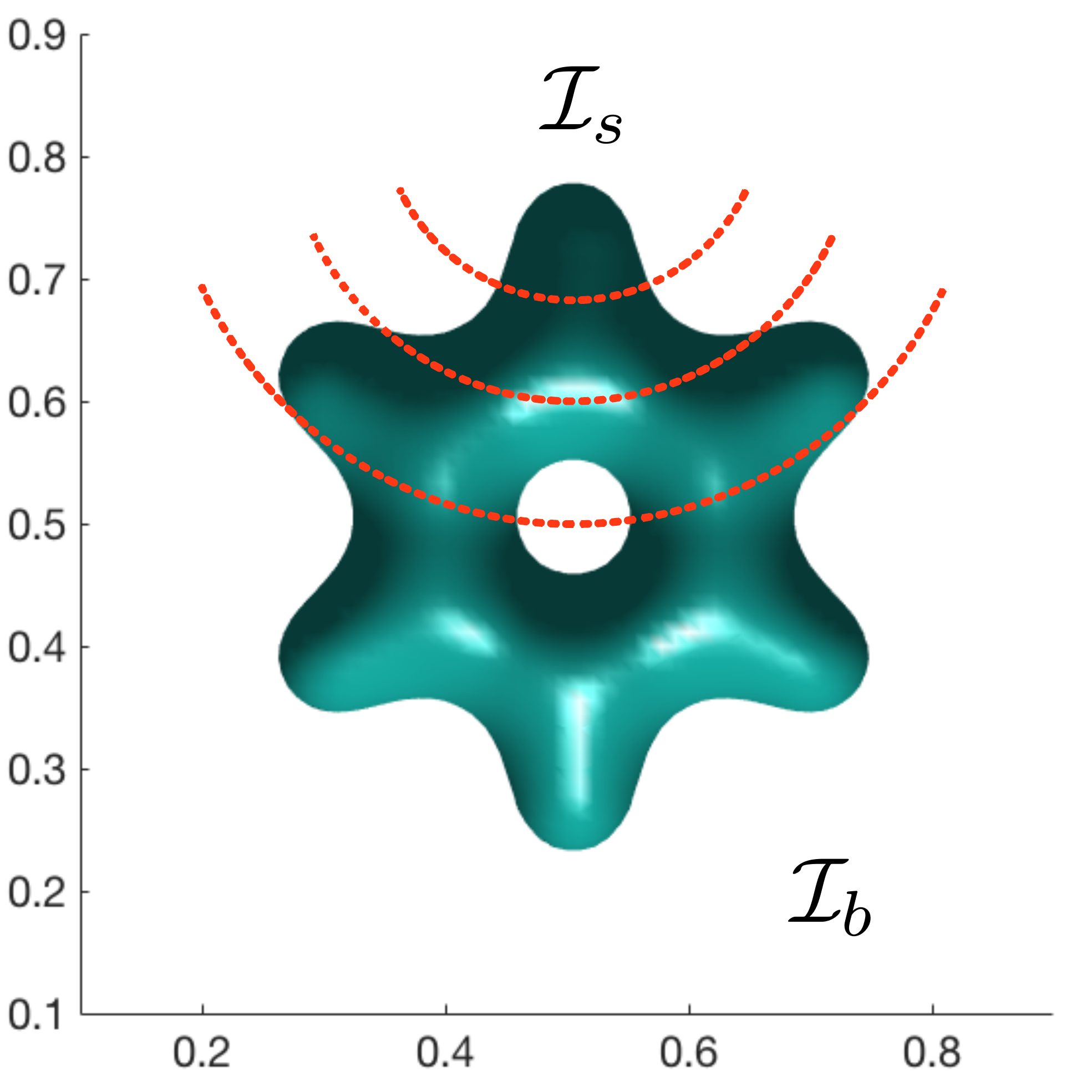}}
\subfloat[]{\includegraphics[trim= 0mm 0mm 0mm 0mm,clip, width=6.0cm]{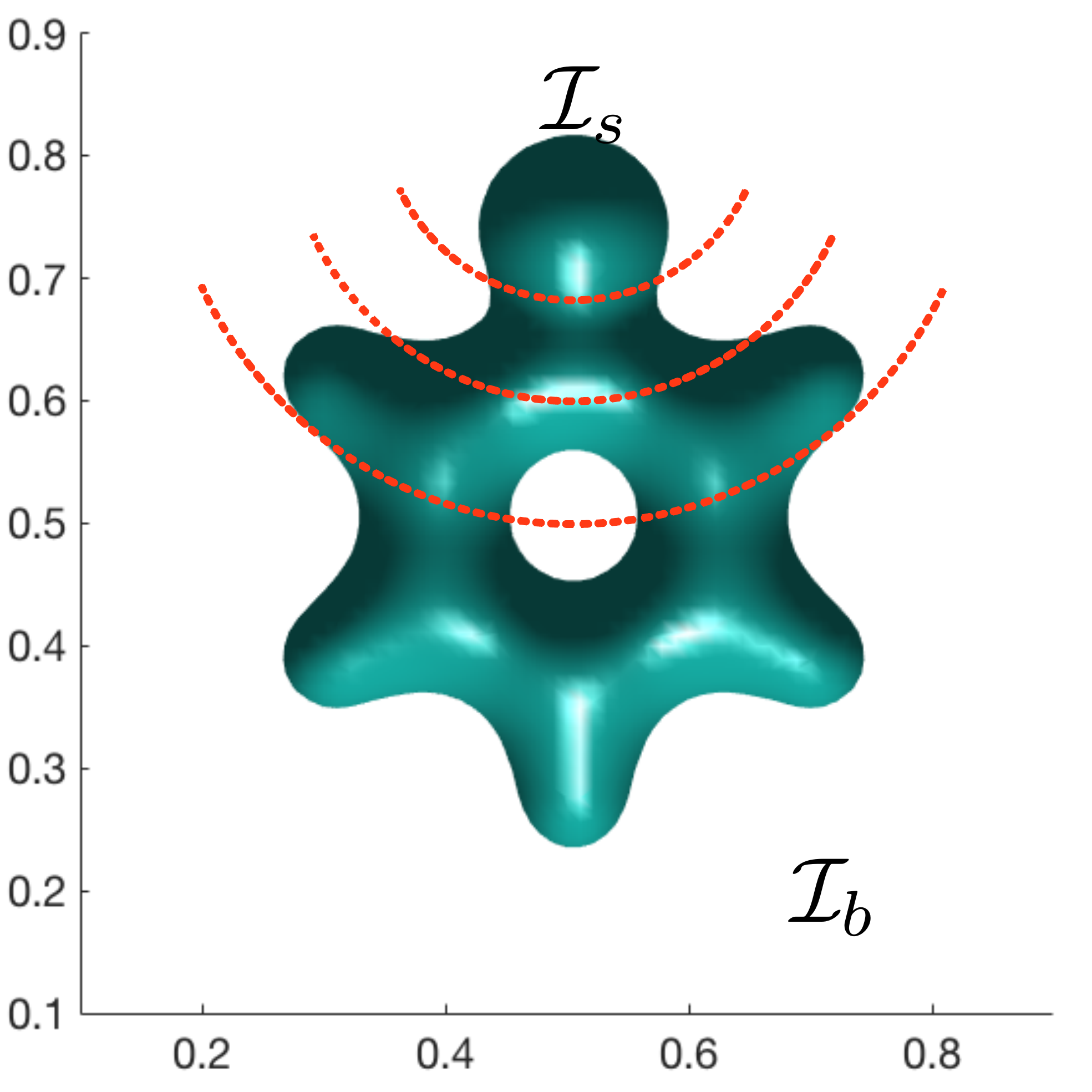}}
\caption{(a) Electron density for the benzene molecule,
and (b) the benzene molecule with an hydrogen atom replace by a fluorine
.} \label{fig:rho_bencene}
\end{figure}

\begin{table}
    \begin{center}
        \begin{tabular}{|c|c|c|c|c|}
            \hline
            $N_{s}$ & $ \frac{|| P - P^{\text{PET}}||}{||P||}$
                    & $ \frac{|| P - P^{\scriptsize{\mbox{pert}}} ||}{||P||}$
                    & $ \frac{||\rho^{\text{PET}} - \rho||}{||\rho||} $
                    & $ \frac{||\rho^{\scriptsize{\mbox{pert}}} - \rho||}{||\rho||} $ \\
            \hline
            1   & $9.41\times 10^{-2}$  &  $5.65\times 10^{-2}$ & $1.36\times 10^{-2}$  &  $2.04\times 10^{-2}$  \\
            3   & $6.32\times 10^{-2}$  &  $1.93\times 10^{-2}$ & $5.56\times 10^{-3}$  &  $6.70\times 10^{-3}$  \\
            5   & $6.46\times 10^{-2}$  &  $1.61\times 10^{-2}$ & $4.41\times 10^{-3}$  &  $3.27\times 10^{-3}$  \\
            7   & $5.04\times 10^{-2}$  &  $1.13\times 10^{-2}$ & $2.93\times 10^{-3}$  &  $1.72\times 10^{-3}$  \\
            9   & $2.98\times 10^{-2}$  &  $4.12\times 10^{-3}$ & $1.56\times 10^{-3}$  &  $1.32\times 10^{-3}$  \\
%            11  & $2.86\times 10^{-2}$  &  $4.74\times 10^{-3}$ & $1.06\times 10^{-3}$  &  $6.67\times 10^{-4}$  \\
            \hline
        \end{tabular}
    \end{center}
    \caption{Errors of the density matrix and electron density for the
    benzene molecule for different bath (and system) sizes.}   \label{table:numerical_results_P_rho_bencene}
    \vspace{-.3cm}
\end{table}

\begin{table}
    \begin{center}
        \begin{tabular}{|c|c|c|c|c|}
            \hline
            $N_{s}$  & $ E - E^{\text{PET}}$
                     & $ E - E^{\scriptsize{\mbox{pert}}}$
                     & $ F - F^{\text{PET}}$
                     & $ F - F^{\scriptsize{\mbox{pert}}}$  \\
            \hline
            1   &  $4.05\times 10^{-2}$  &  $1.33\times 10^{-2}$ & $2.69\times 10^{-2}$  & $3.16\times 10^{-2}$  \\
            3   &  $1.87\times 10^{-2}$  &  $3.42\times 10^{-3}$ & $1.73\times 10^{-2}$  & $9.98\times 10^{-3}$  \\
            5   &  $1.20\times 10^{-2}$  &  $1.78\times 10^{-3}$ & $4.24\times 10^{-3}$  & $6.78\times 10^{-3}$  \\
            7   &  $7.89\times 10^{-3}$  &  $1.02\times 10^{-4}$ & $3.73\times 10^{-3}$  & $4.17\times 10^{-4}$  \\
            9   &  $3.02\times 10^{-3}$  &  $3.16\times 10^{-5}$ & $4.05\times 10^{-3}$  & $3.31\times 10^{-4}$  \\
            11  &  $2.81\times 10^{-3}$  &  $6.71\times 10^{-5}$ & $3.68\times 10^{-3}$  & $3.50\times 10^{-4}$  \\
            \hline
        \end{tabular}
    \end{center}
    \caption{Errors of the energy and forces for the benzene molecule for different bath (and system) sizes.}   \label{table:numerical_results_E_F_bencene}
    \vspace{-.3cm}
\end{table}

\subsubsection{Anthracene}

Finally we test our algorithm with the anthracene molecule
(C$_{14}$H$_{10}$), which is composed of $3$ benzene rings positioned
longitudinally. Following the same procedure as with
the benzene molecule, we compute the solution to the Kohn-Sham equations, whose
\REV{electron} density is shown in Fig.~\ref{fig:rho_anthracene} , and we replace one
hydrogen atom in one of the extremal rings by a fluorine atom (Fig.~\ref{fig:rho_anthracene} (b)). From the
total 33 orbitals for the anthracene, we define the bath orbitals and systems orbitals
following the \REV{same procedure as for the benzene molecule}
The partitions for $N_s = 1, 4,$ and $6$ are depicted in Fig.~\ref{fig:rho_anthracene},
where the different segmented red lines indicate the boundary between the two
partitions, in which they are denoted by $\mathcal{I}_b$ and $\mathcal{I}_s$,
for the bath and for the system respectively.

We compute the PET approximation and its perturbative correction for several
different bath sizes as shown in
Table~\ref{table:numerical_results_P_rho_anthracene}. From
Table~\ref{table:numerical_results_P_rho_anthracene} we can clearly observe
the error of all quantities decrease systematically with respect to the
increase of the system size, and the perturbation method significantly
increases the accuracy over the PET results. In particular, when the
\REV{system} size is $7$, chemical accuracy is achieved after the perturbative
correction is applied.

\begin{figure}
\centering
\subfloat[]{\includegraphics[trim= 0mm 0mm 0mm 0mm,clip, width=6.0cm]{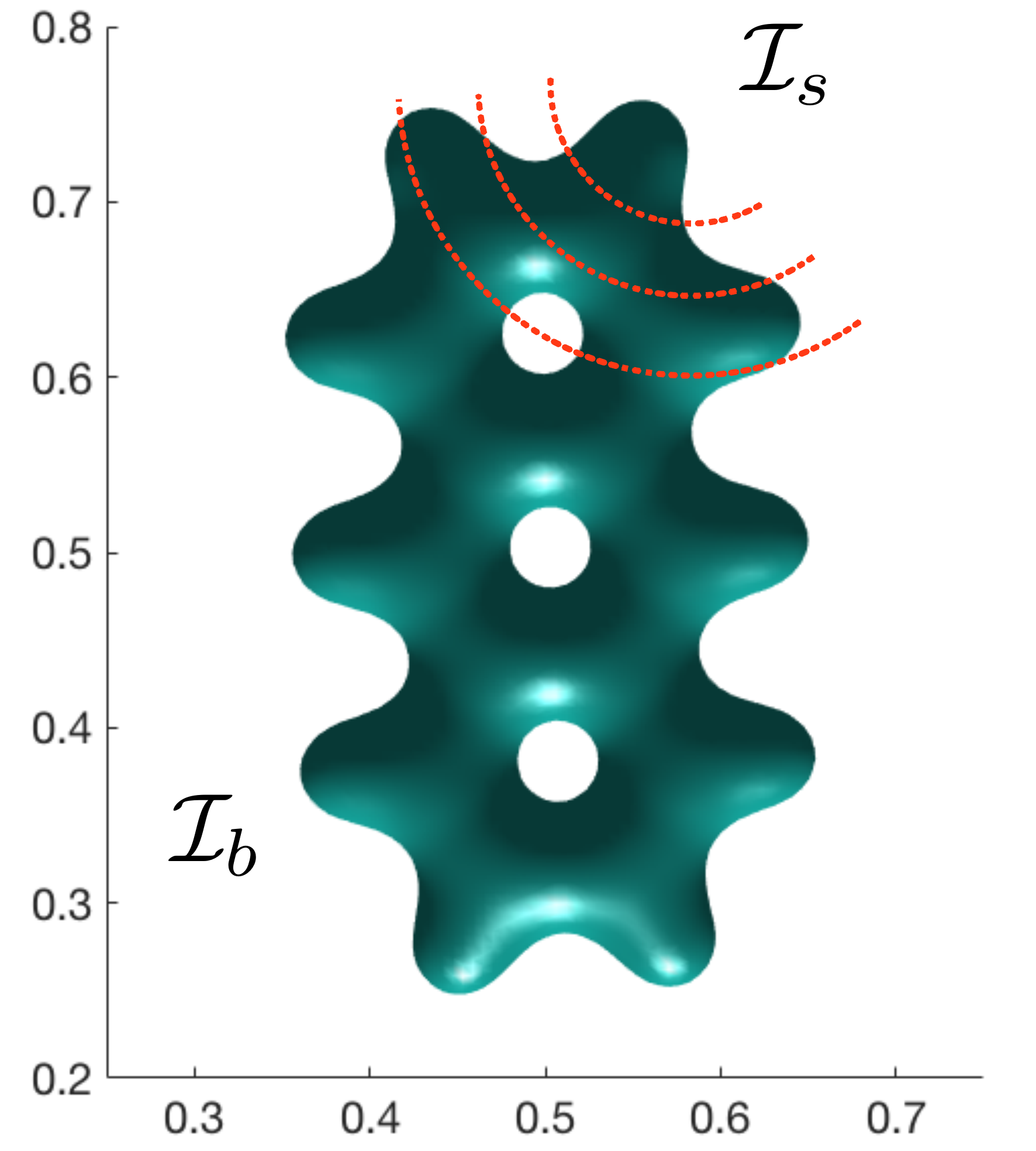}}
\subfloat[]{\includegraphics[trim= 0mm 0mm 0mm 0mm,clip, width=6.0cm]{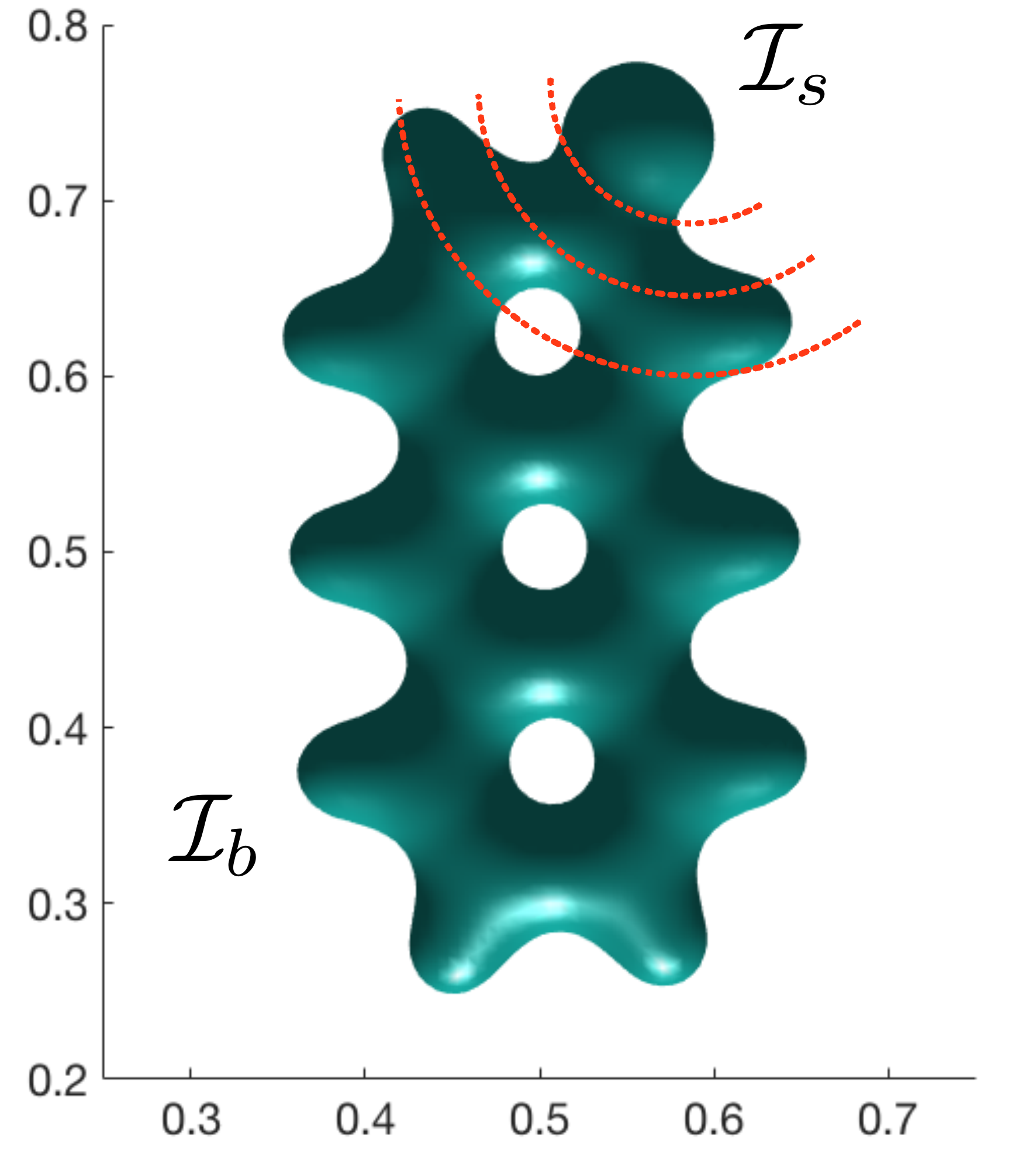}}
\caption{(a) Electron density for the anthracene molecule, and (b) the anthracene molecule with an hydrogen atom replaced by a fluorine.} \label{fig:rho_anthracene}
\end{figure}

\begin{table}
    \begin{center}
        \begin{tabular}{|c|c|c|c|c|}
        \hline
            $N_{s}$ & $ \frac{|| P - P^{\text{PET}}||}{||P||} $
                    & $ \frac{|| P - P^{\scriptsize{\mbox{pert}}} ||}{||P||}$
                    & $ \frac{||\rho^{\text{PET}} - \rho||}{||\rho||} $
                    & $ \frac{||\rho^{\scriptsize{\mbox{pert}}} - \rho||}{||\rho||} $ \\
            \hline
            1  & $8.65\times 10^{-2}$ &  $4.85\times 10^{-2}$ & $1.59\times 10^{-2}$ &  $1.50\times 10^{-2}$  \\
            3  & $6.10\times 10^{-2}$ &  $2.07\times 10^{-2}$ & $5.47\times 10^{-3}$ &  $7.18\times 10^{-3}$   \\
            5  & $5.01\times 10^{-2}$ &  $3.33\times 10^{-2}$ & $3.24\times 10^{-3}$ &  $3.25\times 10^{-3}$  \\
            7  & $4.42\times 10^{-2}$ &  $1.61\times 10^{-2}$ & $2.47\times 10^{-3}$ &  $1.91\times 10^{-3}$  \\
            9  & $3.31\times 10^{-2}$ &  $6.33\times 10^{-3}$ & $1.13\times 10^{-3}$ &  $9.75\times 10^{-4}$  \\
            11 & $2.88\times 10^{-2}$ &  $6.31\times 10^{-3}$ & $1.12\times 10^{-3}$ &  $9.04\times 10^{-4}$  \\
            13 & $2.86\times 10^{-2}$ &  $6.63\times 10^{-3}$ & $1.03\times 10^{-3}$ &  $7.90\times 10^{-4}$  \\
            15 & $1.83\times 10^{-2}$ &  $4.40\times 10^{-3}$ & $6.73\times 10^{-4}$ &  $5.68\times 10^{-4}$  \\
            19 & $1.73\times 10^{-2}$ &  $3.12\times 10^{-3}$ & $5.23\times 10^{-4}$ &  $4.13\times 10^{-4}$  \\
            % 21 & $1.80\times 10^{-2}$ &  $3.68\times 10^{-3}$ & $3.83\times 10^{-4}$ &  $2.34\times 10^{-3}$ &   $9.29\times 10^{-5}$ &  $9.26\times 10^{-3}$ &  $3.42\times 10^{-2}$ \\
            % 23 & $1.29\times 10^{-2}$ &  $2.38\times 10^{-3}$ & $3.81\times 10^{-4}$ &  $1.03\times 10^{-2}$ &   $8.14\times 10^{-5}$ &  $8.68\times 10^{-3}$ &  $3.16\times 10^{-2}$ \\
            % 25 & $1.06\times 10^{-2}$ &  $2.09\times 10^{-3}$ & $2.81\times 10^{-4}$ &  $1.98\times 10^{-3}$ &   $3.09\times 10^{-5}$ &  $6.44\times 10^{-3}$ &  $2.04\times 10^{-2}$ \\
            % 27 & $1.04\times 10^{-2}$ &  $2.01\times 10^{-3}$ & $2.04\times 10^{-4}$ &  $1.64\times 10^{-3}$ &   $1.93\times 10^{-5}$ &  $5.22\times 10^{-3}$ &  $1.59\times 10^{-2}$ \\
            % 29 & $1.01\times 10^{-2}$ &  $1.79\times 10^{-3}$ & $1.96\times 10^{-4}$ &  $5.23\times 10^{-4}$ &   $2.06\times 10^{-5}$ &  $4.13\times 10^{-3}$ &  $1.59\times 10^{-2}$ \\
            \hline
        \end{tabular}
    \end{center}
    \caption{ Errors for the anthracene molecule for different bath (and system) sizes.}   \label{table:numerical_results_P_rho_anthracene}
    \vspace{-.3cm}
\end{table}

\begin{table}
    \begin{center}
        \begin{tabular}{|c|c|c|c|c|}
            \hline
            $N_{s}$  & $ E - E^{\text{PET}}$
                     & $ E - E^{\scriptsize{\mbox{pert}}}$
                     & $ F - F^{\text{PET}}$
                     & $ F - F^{\scriptsize{\mbox{pert}}}$  \\
            \hline
            1   &  $4.06\times 10^{-2}$  &  $1.33\times 10^{-2}$ & $2.69\times 10^{-2}$  & $3.16\times 10^{-2}$  \\
            3   &  $1.87\times 10^{-2}$  &  $3.42\times 10^{-3}$ & $1.73\times 10^{-2}$  & $9.98\times 10^{-3}$  \\
            5   &  $1.10\times 10^{-2}$  &  $1.97\times 10^{-3}$ & $7.72\times 10^{-3}$  & $6.78\times 10^{-3}$  \\
            7   &  $8.18\times 10^{-3}$  &  $2.02\times 10^{-4}$ & $3.73\times 10^{-3}$  & $9.79\times 10^{-4}$  \\
            9   &  $4.02\times 10^{-3}$  &  $3.16\times 10^{-5}$ & $4.05\times 10^{-3}$  & $3.31\times 10^{-4}$  \\
            11  &  $2.74\times 10^{-3}$  &  $2.53\times 10^{-5}$ & $4.42\times 10^{-3}$  & $2.64\times 10^{-4}$  \\
            13  &  $2.26\times 10^{-3}$  &  $6.35\times 10^{-5}$ & $4.53\times 10^{-3}$  & $2.15\times 10^{-4}$  \\
            15  &  $9.70\times 10^{-4}$  &  $1.15\times 10^{-5}$ & $1.84\times 10^{-3}$  & $2.83\times 10^{-4}$  \\
            19  &  $7.34\times 10^{-4}$  &  $5.16\times 10^{-6}$ & $1.62\times 10^{-3}$  & $1.02\times 10^{-4}$  \\
            \hline
        \end{tabular}
    \end{center}
    \caption{ Errors of the energy and forces for the anthracene molecule for different bath (and system) sizes.}   \label{table:numerical_results_E_F_anthracene}
    \vspace{-.3cm}
\end{table}

%\end{comment}

\section{Conclusion}\label{sec:conclusion}

We have studied the recently developed projection based embedding theory
(PET) from a mathematical perspective.  Viewed as a method to
approximately solve eigenvalue problems, PET solves a deflated
eigenvalue problem by taking into account the knowledge from a related
reference system.  This deflated eigenvalue problem can be derived from
the Euler-Lagrange equation of a standard energy minimization procedure
with respect to the density matrices, by with a non-standard constraint
on the \REV{feasible set}. From this perspective, the original formulation of PET
can be seen as a penalty method for imposing the constraint.  Numerical
examples for linear problems as well as nonlinear problems from
Kohn-Sham density functional theory calculations indicate that PET can
yield accurate approximation to the density matrix, energy and atomic
forces.  In order to further improve the accuracy of PET, we
developed a first order perturbation formula. We find that with the help
of the perturbative treatment, PET can achieve chemical accuracy
even for systems of relatively small sizes.

There are several immediate directions for future work. First, we have
studied PET when the system and bath are treated using the same level of
theory. From a physics perspective, it is more attractive to consider
the case when the system part is treated with a more accurate theory
than KSDFT with semi-local exchange-correlation functionals. In particular, it would be
interesting to understand PET when the system part is treated using
KSDFT with nonlocal functionals such as hybrid functionals, or
wavefunction theories such as the coupled cluster (CC) method. It is
also interesting to explore the PET in the context of solving
time-dependent problems. Second,
PET provides a size consistent alternative for many methods in quantum
physics and chemistry to be applied to solid state systems. Some
directions have already been pursued recently for using PET in the
context of periodic
systems~\cite{ChulhaiGoodpaster2018,LibischMarsmanBurgdorferEtAl2017}.
Third, the computation of the atomic force in PET is currently performed
using the finite difference formula, which is expensive in practice. It
would be desirable to develop a method with cost comparable to the
Hellmann-Feynman method but without significant sacrifice of the accuracy.
We note that there has been recent progress along this direction \cite{millerPETForcesArxiv}.
Finally, we believe that the asymptotic convergence property of PET is still dictated by the
nearsightedness principle for systems satisfying the gap condition, but
numerical results indicate that PET already achieves high accuracy even
for system sizes that are well below the prediction from localization
theories. Therefore it is worthwhile to further study the convergence
properties of PET, as well as to perform further comparison with linear scaling type
methods.

\section*{Acknowledgment}

This work was partially supported by the Department of Energy under
Grant No.~DE-SC0017867, No.~DE-AC02-05CH11231, the SciDAC program, and
by the Air Force Office of Scientific Research under award number
FA9550-18-1-0095 (L.~L. and L.~Z.-N.), and by the National Science
Foundation under Grant No. DMS-1652330 (L.~L.).  We thank the Berkeley
Research Computing (BRC) program at the University of California,
Berkeley for making computational resources available.  We thank Garnet
Chan and Frederick Manby for discussions, and Joonho Lee for
valuable suggestions and careful reading of the manuscript.

\bibliographystyle{siam}
\bibliography{pet}

\begin{thebibliography}{10}

\bibitem{BaroniGironcoliDalEtAl2001}
{\sc S.~Baroni, S.~de~Gironcoli, A.~Dal~Corso, and P.~Giannozzi}, {\em Phonons
  and related crystal properties from density-functional perturbation theory},
  Rev. Mod. Phys., 73 (2001), pp.~515--562.

\bibitem{BenziBoitoRazouk2013}
{\sc M.~Benzi, P.~Boito, and N.~Razouk}, {\em Decay properties of spectral
  projectors with applications to electronic structure}, SIAM Rev., 55 (2013),
  pp.~3--64.

\bibitem{BernholcLipariPantelides1978}
{\sc J.~Bernholc, N.~O. Lipari, and S.~T. Pantelides}, {\em Self-consistent
  method for point defects in semiconductors: Application to the vacancy in
  silicon}, Phys. Rev. Lett., 41 (1978), p.~895.

\bibitem{BrouderPanatiCalandraEtAl2007}
{\sc C.~Brouder, G.~Panati, M.~Calandra, C.~Mourougane, and N.~Marzari}, {\em
  Exponential localization of {Wannier} functions in insulators}, Phys. Rev.
  Lett., 98 (2007), p.~046402.

\bibitem{ChenOrtner2015}
{\sc H.~Chen and C.~Ortner}, {\em {QM/MM methods for crystalline defects. Part
  2: Consistent energy and force-mixing}}, arXiv:1509.06627,  (2015).

\bibitem{ChibaniRenSchefflerEtAl2016}
{\sc W.~Chibani, X.~Ren, M.~Scheffler, and P.~Rinke}, {\em {Self-consistent
  Green's function embedding for advanced electronic structure methods based on
  a dynamical mean-field concept}}, Phys. Rev. B, 93 (2016), p.~165106.

\bibitem{ChulhaiGoodpaster2018}
{\sc D.~V. Chulhai and J.~D. Goodpaster}, {\em Projection-based correlated wave
  function in density functional theory embedding for periodic systems}, J.
  Chem. Theory Comput., 14 (2018), pp.~1928--1942.

\bibitem{Cortona1991}
{\sc P.~Cortona}, {\em Self-consistently determined properties of solids
  without band-structure calculations}, Phys. Rev. B, 44 (1991), p.~8454.

\bibitem{DamleLinYing2015}
{\sc A.~Damle, L.~Lin, and L.~Ying}, {\em Compressed representation of
  {K}ohn--{S}ham orbitals via selected columns of the density matrix}, J. Chem.
  Theory Comput., 11 (2015), pp.~1463--1469.

\bibitem{DamleLinYing2017a}
{\sc A.~Damle, L.~Lin, and L.~Ying}, {\em Accelerating selected columns of the
  density matrix computations via approximate column selection}, SIAM J. Sci.
  Comput., 39 (2017), p.~1178.

\bibitem{FosterBoys1960}
{\sc J.~M. Foster and S.~F. Boys}, {\em Canonical configurational interaction
  procedure}, Rev. Mod. Phys., 32 (1960), p.~300.

\bibitem{GarciaLuE2007}
{\sc C.J. Garc{\'\i}a-Cervera, J.~Lu, and W.~E}, {\em Asymptotics-based
  sub-linear scaling algorithms and application to the study of the electronic
  structure of materials}, Commun. Math. Sci., 5 (2007), pp.~999--1024.

\bibitem{GolubVan2013}
{\sc G.~H. Golub and C.~F. Van~Loan}, {\em Matrix computations}, Johns Hopkins
  Univ. Press, Baltimore, fourth~ed., 2013.

\bibitem{GoodpasterAnanthManbyEtAl2010}
{\sc J.~D. Goodpaster, N.~Ananth, F.~R. Manby, and T.~F. Miller~III}, {\em
  Exact nonadditive kinetic potentials for embedded density functional theory},
  J. Chem. Phys., 133 (2010), p.~084103.

\bibitem{HohenbergKohn1964}
{\sc P.~Hohenberg and W.~Kohn}, {\em {Inhomogeneous electron gas}}, Phys. Rev.,
  136 (1964), pp.~B864--B871.

\bibitem{HuangPavoneCarter2011}
{\sc C.~Huang, M.~Pavone, and E.~A. Carter}, {\em Quantum mechanical embedding
  theory based on a unique embedding potential}, J. Chem. Phys., 134 (2011),
  p.~154110.

\bibitem{Huzinaga1971}
{\sc S.~Huzinaga and A.~A. Cantu}, {\em Theory of separability of many electron
  systems}, J. Chem. Phys., 55 (1971), pp.~5543--5549.

\bibitem{KananenkaGullZgid2015}
{\sc A.~A. Kananenka, E.~Gull, and D.~Zgid}, {\em Systematically improvable
  multiscale solver for correlated electron systems}, Phys. Rev. B, 91 (2015),
  p.~121111.

\bibitem{KellyCar1992}
{\sc P.~J. Kelly and R.~Car}, {\em Green's-matrix calculation of total energies
  of point defects in silicon}, Phys. Rev. B, 45 (1992), p.~6543.

\bibitem{KniziaChan2013}
{\sc G.~Knizia and G.~K.-L. Chan}, {\em Density matrix embedding: A
  strong-coupling quantum embedding theory}, J. Chem. Theory Comput., 9 (2013),
  pp.~1428--1432.

\bibitem{Knyazev2001}
{\sc A.~V. Knyazev}, {\em Toward the optimal preconditioned eigensolver:
  Locally optimal block preconditioned conjugate gradient method}, SIAM J. Sci.
  Comp., 23 (2001), pp.~517--541.

\bibitem{Kohn1996}
{\sc W.~Kohn}, {\em Density functional and density matrix method scaling
  linearly with the number of atoms}, Phys. Rev. Lett., 76 (1996),
  pp.~3168--3171.

\bibitem{KohnSham1965}
{\sc W.~Kohn and L.~Sham}, {\em {Self-consistent equations including exchange
  and correlation effects}}, Phys. Rev., 140 (1965), pp.~A1133--A1138.

\bibitem{millerPETForcesArxiv}
{\sc S.J.R. Lee, F.~Ding, F.R. Manby, and T.F. Miller~III}, {\em {Analytical
  Gradients for Projection-Based Wavefunction-in-DFT Embedding}},
  arXiv:1903.05830,  (2019).

\bibitem{LiLinLu2018}
{\sc X.~Li, L.~Lin, and J.~Lu}, {\em {PEXSI}-$\sigma$: {A Green's function
  embedding method for Kohn-Sham density functional theory}}, Ann. Math. Sci.
  Appl., 3 (2018), p.~411.

\bibitem{LibischMarsmanBurgdorferEtAl2017}
{\sc F.~Libisch, M.~Marsman, J.~Burgd{\"o}rfer, and G.~Kresse}, {\em Embedding
  for bulk systems using localized atomic orbitals}, J. Chem. Phys., 147
  (2017), p.~034110.

\bibitem{LiuWenWangEtAl2015}
{\sc X.~Liu, Z.~Wen, X.~Wang, M.~Ulbrich, and Y.~Yuan}, {\em On the analysis of
  the discretized {Kohn--Sham} density functional theory}, SIAM J. Numer.
  Anal., 53 (2015), pp.~1758--1785.

\bibitem{ManbyStellaGoodpasterEtAl2012}
{\sc F.~R. Manby, M.~Stella, J.~D. Goodpaster, and T.~F. Miller~III}, {\em {A
  simple, exact density-functional-theory embedding scheme}}, J. Chem. Theory
  Comput., 8 (2012), pp.~2564--2568.

\bibitem{MarzariVanderbilt1997}
{\sc N.~Marzari and D.~Vanderbilt}, {\em Maximally localized generalized
  {W}annier functions for composite energy bands}, Phys. Rev. B, 56 (1997),
  p.~12847.

\bibitem{NguyenKananenkaZgid2016}
{\sc T.~Nguyen, A.~A. Kananenka, and D.~Zgid}, {\em Rigorous ab initio quantum
  embedding for quantum chemistry using {Green's} function theory: Screened
  interaction, nonlocal self-energy relaxation, orbital basis, and chemical
  accuracy}, J. Chem. Theory Comput.,  (2016).

\bibitem{ProdanKohn2005}
{\sc E.~Prodan and W.~Kohn}, {\em {Nearsightedness of electronic matter}},
  Proc. Natl. Acad. Sci., 102 (2005), pp.~11635--11638.

\bibitem{SaadSchultz1986}
{\sc Y.~Saad and M.~H. Schultz}, {\em {GMRES}: {A} generalized minimal residual
  algorithm for solving nonsymmetric linear systems}, SIAM J. Sci. Stat.
  Comput., 7 (1986), pp.~856--869.

\bibitem{SunChan2016}
{\sc Q.~Sun and G.~K.-L. Chan}, {\em Quantum embedding theories}, Acc. Chem.
  Res., 49 (2016), pp.~2705--2712.

\bibitem{TeterPayneAllan1989}
{\sc M.P. Teter, M.C. Payne, and D.C. Allan}, {\em {Solution of
  Schr{\"o}dinger's equation for large systems}}, Phys. Rev. B, 40 (1989),
  p.~12255.

\bibitem{WilliamsFeibelmanLang1982}
{\sc A.~R. Williams, P.~J. Feibelman, and N.~D. Lang}, {\em Green's-function
  methods for electronic-structure calculations}, Phys. Rev. B, 26 (1982),
  p.~5433.

\bibitem{YangMezaLeeEtAl2009}
{\sc C.~Yang, J.~C. Meza, B.~Lee, and L.~W. Wang}, {\em {KSSOLV--a MATLAB
  toolbox for solving the Kohn--Sham equations}}, ACM Trans. Math. Software, 36
  (2009), p.~10.

\bibitem{ZellerDederichs1979}
{\sc R.~Zeller and P.~H. Dederichs}, {\em {Electronic Structure of Impurities
  in Cu, Calculated Self-Consistently by Korringa-Kohn-Rostoker
  Green's-Function Method}}, Phys. Rev. Lett., 42 (1979), p.~1713.

\bibitem{ZgidChan2011}
{\sc D.~Zgid and G.~K.-L. Chan}, {\em {Dynamical mean-field theory from a
  quantum chemical perspective}}, J. Chem. Phys., 134 (2011), p.~094115.

\end{thebibliography}

\end{document}